\documentclass[journal,twoside,web]{ieeecolor}
\usepackage{generic}
\usepackage{cite}
\usepackage{amsmath,amssymb,amsfonts}
\usepackage{algorithmic}
\usepackage{graphicx}
\usepackage{textcomp}
\def\BibTeX{{\rm B\kern-.05em{\sc i\kern-.025em b}\kern-.08em
    T\kern-.1667em\lower.7ex\hbox{E}\kern-.125emX}}
\let\labelindent\relax
\usepackage{enumitem,boondox-cal}
\usepackage{subfloat,subfigure}

\usepackage{amsthm}
\usepackage[T1]{fontenc}

\usepackage[
noabbrev,
capitalise,
nameinlink,
]{cleveref}
\crefname{lemma}{Lemma}{Lemmas}
\crefname{lemmax}{Lemma}{Lemmas}
\crefname{theorem}{Theorem}{Theorems}
\crefname{theoremx}{Theorem}{Theorems}
\crefname{assumption}{Assumption}{Assumptions}
\crefname{proposition}{Proposition}{Propositions}
\crefname{propositionx}{Proposition}{Propositions}
\crefname{corollary}{Corollary}{Corollaries}
\crefname{appen}{Appendix}{Appendices}
\crefname{figure}{Figure}{Figures}
\crefname{algorithm}{Algorithm}{Algorithms}
\crefname{remark}{Remark}{Remarks}

\theoremstyle{remark} 
\newtheorem{theorem}{Theorem}

\newtheorem{lemma}{Lemma}
\newtheorem{theoremx}{Theorem}[section]

\newtheorem{lemmax}{Lemma}[section]
\newtheorem{corollary}{Corollary}[section]
\newtheorem{assumption}{Assumption}

\theoremstyle{remark} 
\newtheorem{remark}{Remark}
\newtheorem{remarkx}{Remark}[section]

\def\mP{\mathbb{P}}
\def\mE{\mathbb{E}}
\def\td{\text{d}}

\newcommand{\abs}[1]{\left| #1\right|}
\newcommand{\Abs}[1]{\left\| #1\right\|}

\def\udf{\underline{f}}
\def\epp{\varepsilon_1}
\def\Mpri{M^\prime}

\begin{document}
\title{Recursive Identification of Binary-Valued Systems Under Uniform Persistent Excitations}
\author{Jieming Ke, \IEEEmembership{Student Member, IEEE}, Ying Wang, \IEEEmembership{Member, IEEE},  Yanlong Zhao, \IEEEmembership{Senior Member, IEEE}, and Ji-Feng Zhang, \IEEEmembership{Fellow, IEEE}
\thanks{The work is supported by National Key R\&D Program of China under Grant 2018YFA0703800, National Natural Science Foundation of China under Grants 62025306, 62303452 and T2293770, CAS Project for Young Scientists in Basic Research under Grant YSBR-008, and China Postdoctoral Science Foundation under Grant 2022M720159.}
\thanks{Jieming Ke, Ying Wang, Yanlong Zhao, and Ji-Feng Zhang are with the Key Laboratory of Systems and Control, Institute of Systems Science, Academy of Mathematics and Systems Science, Chinese Academy of Sciences, Beijing 100190, China, and	also with the School of Mathematics Sciences, University of Chinese	Academy of Sciences, Bejing 100149, China (e-mail: 
	kejieming@amss.ac.cn; 
	wangying96@amss.ac.cn; 
	ylzhao@amss.ac.cn; 
	jif@iss.ac.cn). }}

\maketitle

\begin{abstract}
This paper investigates the online identification problem of {binary-valued} moving average systems. 
A stochastic approximation-based algorithm without projections or truncations is proposed. 
To analyze the convergence property of the algorithm, the distribution tail of the parameter estimate is proved to be exponentially convergent through an auxiliary stochastic process.
Under uniform persistent excitations, the almost sure and mean square convergence of the algorithm is obtained. 
When the step-size coefficient is properly selected, the almost sure and mean square convergence rates are proved to reach $ O(\sqrt{\ln \ln k/k}) $ and $ O(1/k) $ respectively, where $ k $ is the sample size.
A numerical example is given to demonstrate the effectiveness of the proposed algorithm and theoretical results.
\end{abstract}


\section{Introduction}

{Binary-valued} systems emerge widely in practice.
For example, in automotive and chemical process applications, oxygen sensors are used for evaluating gas oxygen contents \cite{Wang2002prediction,Wang2003System,WangLY2010System}. Inexpensive oxygen sensors are switching types that change their voltage outputs sharply when excess oxygen in the gas is detected. 
More examples can be seen in genetic association studies \cite{bi2021efficient,kang2017robust}, radar target recognition \cite{wang2016radar}, and credit scoring \cite{Wang2021credit}, etc. 
The appearance of the above {binary-valued} sensors brings forward new requirements for identification theory,
which is the focus of the paper.

There are some important identification algorithms proposed for binary-valued and finite-valued systems \cite{Bottegal2017a,Colinet2010a,godoy2011on,gustafsson2009statistical,Risuleo2020identification,Shen2014robust,zhao2023system}, many of which are offline. Offline methods take full advantage of the statistical property of the finite-valued outputs, and require fewer assumptions than the online ones. However, in some scenarios, for instance, in adaptive control problems, online identification is of great importance, since online identification methods need less memory and computation complexity, and can update the parameter estimate quickly \cite{ljung2}.




The online identification of binary-valued and finite-valued systems has
been investigated under different type inputs \cite{Csaji2012recursive,Song2018recursive,Song2024identification,You2015recursive,Diao2020a,He2013Moderate,Mei2014almost,Moschitta2015Parametric,Wang2007asymptotically,WangLY2006Joint,WangLY2008Identification,Wang2003System,Zhao2010Identification,zhao2018consensus,Fu2022Distributed,WangY2022unified,guo2013recursive,zhang2019asymptotically,Wang2021distributed}. For example, \cite{Csaji2012recursive,Song2018recursive,Song2024identification} assume the inputs to be independent and identically distributed (i.i.d.), and propose stochastic approximation algorithms with expanding truncations for {binary-valued} systems.  \cite{You2015recursive} requires i.i.d. inputs, and gives a stochastic gradient-based algorithm. These algorithms are all proved to be convergent almost surely. Besides, \cite{Diao2020a,He2013Moderate,Li2021suboptimal,Mei2014almost,Moschitta2015Parametric,Wang2007asymptotically,WangLY2006Joint,WangLY2008Identification,Wang2003System,Zhao2010Identification,zhao2018consensus} consider periodic inputs, and propose empirical measurement methods. The methods can be applied in infinite impulse response systems and Hammerstein systems with {binary-valued} observations \cite{Mei2014almost,Zhao2010Identification}.  \cite{Fu2022Distributed} and \cite{WangY2022unified} consider uniformly persistently exciting inputs, and design sign-error type identification algorithms. \cite{guo2013recursive,zhang2019asymptotically,Wang2021distributed} assume the inputs to be persistently exciting, and propose recursive projection methods. 

There are two types of sensors considered in the finite-valued system identification problems. One type sensors are adaptive ones, whose thresholds can be adjusted according to historical data \cite{Csaji2012recursive,Fu2022Distributed,WangY2022unified}. In the adaptive sensor case, the system outputs provide richer information when the thresholds are properly selected. Another type sensors are fixed ones, whose thresholds are time-invariant \cite{Diao2020a,He2013Moderate,Li2021suboptimal,Mei2014almost,Moschitta2015Parametric,Wang2007asymptotically,WangLY2006Joint,WangLY2008Identification,Wang2003System,Zhao2010Identification,zhao2018consensus,guo2013recursive,zhang2019asymptotically,Wang2021distributed}. Fixed sensors are more common in practical scenarios. A practical example of the fixed finite-valued sensors is the oxygen sensors in automotive and chemical process \cite{Wang2002prediction,Wang2003System,WangLY2010System}.

This paper focuses on the {binary-valued} system identification problem under uniform persistent excitations and fixed {binary-valued} sensors. The problem has been studied in \cite{guo2013recursive,zhang2019asymptotically,Wang2021distributed}, but these works require that the unknown parameter is located in a known compact set. They design projections according to the \textit{a priori} information on parameter location to ensure the uniform boundedness of the identification algorithms. 

We consider the case without any \textit{a priori} information on the parameter location. In this case, {the identification algorithm should have the ability to search unknown parameters in the whole space. Therefore, }the projections in \cite{guo2013recursive,zhang2019asymptotically,Wang2021distributed} should be removed for the convergence properties{, which causes the algorithm to lose the uniform boundedness. This makes it difficult to analyze the convergence properties of the algorithm. To overcome the difficulty, in the periodic input case, \cite{zhao2018consensus} calculates the distribution tail of the parameter estimate, that is the probability that the parameter estimate exceeds a certain compact set. In the non-periodic input case, the distribution of observation sequences does not maintain periodicity and therefore more complex, which makes the distribution tail of the parameter estimate difficult to be calculated. }


To solve the difficulty, this paper constructs a stochastic process with averaged observations (SPAO), which builds a bridge between the average of the {binary-valued} observations and the algorithm. By SPAO, we can utilize the distribution tail of the observation average to estimate the distribution tail of the algorithm. 

In the paper, a stochastic approximation-based (SA-based) algorithm without projections or truncations is proposed for the {binary-valued} moving average (MA) system identification problem. 
The main contributions of the paper are as follows.

\begin{enumerate}[label={\roman*)}]
	\item 
	A new SA-based identification algorithm without projections is proposed for {binary-valued} MA systems. 
	{Using} this algorithm, we can recursively obtain the unknown parameter under uniform persistent excitations without any \textit{a priori} information on the parameter location. It is the first paper where such a property is derived in the fixed finite-level quantizer and non-periodic deterministic input case. 

	\item The convergence properties of the SA-based identification algorithm are established. To be specific, the almost sure convergence and mean square convergence are induced through the exponential convergence of the estimation error distribution tail. Besides, when the step-size coefficient is properly selected, the almost sure convergence rate is proved to be $ O(\sqrt{\ln \ln k/k}) $, which is firstly achieved among online identification algorithms of stochastic {binary-valued} systems under non-periodic inputs. Moreover, the mean square convergence rate is proved to be $ O(1/k) $, which is the best mean square convergence rate in theory under {binary-valued} observations and even accurate ones.

	\item A new constructive methodology is developed for the convergence analysis of {binary-valued} system identification algorithms.
	Specially, SPAO is constructed to reveal the connection between the average of the {binary-valued} observations and the convergence properties of the algorithm.
	Moreover, the methodology is also shown to be practical for a common class of recursive identification algorithms for {binary-valued} systems, such as the stochastic gradient-based algorithm \cite{You2015recursive} and the quasi-Newton type algorithms \cite{zhang2019asymptotically,Wang2021distributed}.
\end{enumerate}

The rest of the paper is organized as follows. \cref{sec:prob form} formulates the identification problem. \cref{sec:iden algo} proposes an SA-based identification algorithm of {binary-valued} systems. \cref{sec:conv prop} gives the convergence analysis. \cref{subsec:SPAO} constructs an auxiliary stochastic process named SPAO and discusses its property. Based on SPAO, \cref{subsec:DistTailEstm} estimates the distribution tail of the estimation error, and gives the almost sure and mean square convergence. \cref{subsec:ASConvRate} and \cref{subsec:M2ConvRate} analyze the almost sure and mean square convergence rates, respectively. \cref{sec:nume simu} simulates a numerical example to demonstrate the theoretical results. \cref{sec:conclusion} gives concluding remarks and future works. 

\subsection*{Notation}

In the rest of the paper, $\mathbb{N}$, $\mathbb{R}$ and $\mathbb{R}^{n}$ denote the sets of natural numbers, real numbers and $n$-dimensional real vectors, respectively. $I_{\{\cdot\}}$ denotes the indicator function, whose value is 1 if its argument (a formula) is true, and 0, otherwise. $\|x\|$ is the Euclidean norm for vector $x$. $I_n$ is an $n\times n$ identity matrix.  $\lfloor x\rfloor$ is the largest integer that is smaller than or equal to $x\in \mathbb{R}$. The positive part of $x$ is denoted as $x^+=\max\{x,0\}$. For square matrices $ A_l, \ldots, A_k $, denote $ \prod_{i=l}^k A_i = A_k \cdots A_l $ for $ k\geq l $. Relations between two series $a_k$ and $b_k$ are defined as

\begin{enumerate}[label={\roman*)}]
	\item $a_k=O(b_k)$ if $a_k=c_k b_k$ for a bounded $c_k$;
	\item $a_k=o(b_k)$ if $a_k=c_k b_k$ for a $c_k$ that converges to $ 0 $.
\end{enumerate}

\section{Problem formulation}\label{sec:prob form}

Consider the MA system:

\begin{equation}\label{1}
	y_k=\phi_k^\top\theta+d_k,\quad k\geq 1,
\end{equation}
where $ \phi_k = \phi(u_k,u_{k-1},\ldots,u_{k-\bar{n}+1}) \in \mathbb{R}^n $ is a regressed function of inputs $ u_k $ for some $ \bar{n}>0 $, $\theta\in \mathbb{R}^n$ is the unknown parameter, and $d_k$ is the system noise, respectively. The unobserved system output $y_k$ is measured by a {binary-valued} sensor with a fixed threshold $C$, which can be represented by an indicator function
\begin{equation}\label{2}
	s_k=I_{\{y_k\leq C\}}=
	\begin{cases}
		1,\quad y_k\leq C;\\
		0,\quad y_k > C.\\
	\end{cases}
\end{equation}

Our goal is to identify the unknown parameter $\theta$ based on the regressed vector $\phi_k$ and the binary observation $s_k$. 

\begin{assumption}\label{a1}
	The sequence $\{\phi_k, k\geq 1\}$ is bounded, i.e.,
	\begin{equation*}
		\sup\limits_{k\geq 1}\Abs{\phi_k}\leq M<\infty,
	\end{equation*}
	and there exist a positive integer $N\geq n$ and a real number $\delta>0$ such that
	\begin{equation}\label{AssumEq_SuffRichCond}
		\frac{1}{N}\sum^{k+N-1}_{i=k}\phi_i \phi_i^\top\geq\delta I_n,\ k\geq 1.
	\end{equation}
\end{assumption}

\begin{remark}
	The condition \eqref{AssumEq_SuffRichCond} is usually called ``uniform persistent excitation condition'' or ``sufficiently rich condition'' \cite{chen1987adaptive,guo2013recursive}.  \cref{a1} is  common in the {binary-valued} system identification field \cite{guo2013recursive,WangY2022unified}. 
\end{remark}

\begin{assumption}\label{a2}
	The system noise sequence $\{d_k,k\geq 1\}$ is i.i.d. with zero mean and finite variance, whose distribution and density function are  denoted as $F(\cdot)$ and $f(\cdot)$, respectively. The distribution $ F(\cdot) $ is Lipschitz continuous, and the density function $f(\cdot)$ satisfies
	\begin{equation}\label{eq_assump2}
		\inf\limits_{x\in \mathfrak{X}}f(x)>0
	\end{equation}
	for any bounded open set $\mathfrak{X}$.
\end{assumption}

For simplicity of notation, denote
\begin{equation*}
	F_k=F\left(C-\phi_k^\top\theta\right),\quad
	f_k=f\left(C-\phi_k^\top\theta\right).
\end{equation*}
Then $ \mE s_k = \mP\{y_k\leq C\} = F_k $.

\begin{remark}
	Gaussian noise, Laplacian noise and $ t $-distribution noise are all examples satisfying \cref{a2}. Moreover, if \eqref{eq_assump2} does not hold for the system noise, we can add a dither to the binary sensor \cite{Wang2003System}. Under \cref{a2}, the density function $ f(\cdot) $ is bounded because of the Lipschitz continuity of the distribution function $ F(\cdot) $. 
\end{remark}

{
\begin{remark}
	It will be a more general problem when the noise is Gaussian with an unknown variance $ \sigma^2 $. In this case, we can use the similar technique of \cite{WangLY2006Joint} to transform the joint identification problem for $ \theta $ and $ \sigma^2 $ into the identification problem for a new {binary-valued} system with known noise variance. 
\end{remark}
}
%
%


\section{Identification algorithm}\label{sec:iden algo}

This section will propose an SA-based algorithm for the MA system \eqref{1} with binary observation \eqref{2}.

%

In the viewpoint of SA, the identification problem can be treated as the problem to find the roots of 
\begin{align*}
	\mu_k(\hat{\theta}) = F\left(C-\phi_k^\top\hat{\theta}\right) - F_k,\ k \geq 1. 
\end{align*}
Because $ \theta $ is unknown, $ F_k $ is unavailable. Besides, $ s_k $ is available, and its expectation is $ F_k $. We replace $ F_k $ with $ s_k $. 
%
Then, based on the SA method \cite{chen2002stochastic}, the identification algorithm is given as
\begin{align*}
	\hat{\theta}_k
	= \hat{\theta}_{k-1} + \rho_k \phi_k \left( F\left(C-\phi_k^\top\hat{\theta}_{k-1}\right) - s_k \right),
\end{align*}
where $\rho_k\in\mathbb{R}$ is the step-size.   

{
\begin{remark}
	In the algorithm design, $ F_k $ can be replaced with $ s_k $ because $ \{s_k - F_k\} $ is a martingale difference sequence with uniformly bounded variances. When the step-size is properly selected, martingale difference noises with uniformly bounded variances will not influence the convergence of SA-based algorithms \cite{chen2002stochastic}. 
\end{remark}
}

Denote $ \hat{F}_k=F\left(C-\phi_k^\top\hat{\theta}_{k-1}\right) $.
Set $ \rho_k = \beta/k $, where $ \beta>0 $ is a constant coefficient. Then, the SA-based algorithm is given as follows. 
\begin{equation}\label{algo}
	\hat{\theta}_{k}=\hat{\theta}_{k-1}+\frac{\beta \phi_k}{k}\left(\hat{F}_k-s_k\right),\ \forall k>k_0,
\end{equation}
where $ k_0\in\mathbb{N} $ is the starting point, and the initial value $ \hat{\theta}_{k_0} $ can be arbitrarily selected in $ \mathbb{R}^n $.

The observation error $ \hat{\theta}_k - \theta $ is denoted as $ \tilde{\theta}_k $.

\begin{remark}
	Algorithm \eqref{algo} is similar to the recursive projection algorithm proposed in \cite{guo2013recursive}. The main difference is that we do not introduce any projections or truncations in \eqref{algo}. 
	The difference brings major difficulty in the convergence analysis. The convergence analysis of the recursive projection algorithm relies on the fact that if the search region is constrained in a compact set, then there is a uniform positive lower bound for $ -({\hat{F}_k-F_k})/\phi_k^\top\tilde{\theta}_{k-1}  $. Without any projection, Algorithm \eqref{algo} can reach every point in the whole space. Then, the infimum of $ -({\hat{F}_k-F_k})/\phi_k^\top\tilde{\theta}_{k-1} $ can be arbitrarily close to $ 0 $. 
	To overcome the problem, we should investigate the distribution tail of the algorithm. 
\end{remark}

\begin{remark}
		The step-size $ \rho_k $ that converges to $ 0 $ is used to reduce the effect of noise $ d_k $ \cite{chen2002stochastic}. In the SA method \cite{chen2002stochastic}, $ \rho_k $ should satisfy $\sum_{i=1}^\infty \rho_i = \infty$ and $\sum_{i=1}^\infty \rho^{2}_i < \infty$. One of the example is $ \rho_k = \beta/k $ that is used in \eqref{algo}. Another example is $  \rho_k = \beta / (1+\sum_{i=1}^{k}\Abs{\phi_i}^2) $ that is used in \cite{guo2013recursive}.
\end{remark}

\begin{remark}
	In Algorithm \eqref{algo}, $ \hat{F}_k $ is used to approximate $ s_k $ because $ \hat{F}_k = \mE[s_k(\theta)|\theta=\hat{\theta}_{k-1}] $. Therefore, in the multiple threshold case with threshold number $ q $,  Algorithm \eqref{algo} also works after replacing $ \hat{F}_k $ with $ \mE[s_k^q(\theta)|\theta=\hat{\theta}_{k-1}] $, where $ s_k^q $ is the corresponding observation in $ \{0,1,\ldots,q\} $. 
\end{remark}

\section{Convergence}\label{sec:conv prop}

This section will focus on the convergence analysis of the algorithm including the distribution tail, almost sure convergence rate and mean square convergence rate. An auxiliary stochastic process is introduced firstly to assist in the analysis.

\subsection{Stochastic process with averaged observations (SPAO)}\label{subsec:SPAO}

This subsection will introduce an auxiliary stochastic process satisfying
\begin{enumerate}[label=\roman*)]
	\item the trajectory of the stochastic process gradually approaches that of the estimation error $ \tilde{\theta}_k $;
	\item the convergence property of the stochastic process is easy to analyze compared with that of the algorithm.
\end{enumerate}

The construction is inspired by the idea that $ \beta \phi_k (F_k-s_k) $ can be replaced by the linear combination of $w_k$ and $w_{k-1}$, where
\begin{equation}\label{defT}
	w_k = \frac{\sum_{i=1}^k \beta \phi_i (F_i - s_i)}{k},
\end{equation}
i.e.,
\begin{align*}
	\beta \phi_k (F_k-s_k)
	=& \sum_{i=1}^{k} \beta \phi_i (F_i-s_i) - \sum_{i=1}^{k-1} \beta \phi_i (F_i-s_i) \nonumber\\
	=& k\left(w_k - w_{k-1}\right) + w_{k-1}.
\end{align*}
Define $ \psi_k = \tilde{\theta}_k - w_k $. Then, by the transformation above,
\begin{align}\label{recur_psi}
	\psi_k
	= & \psi_{k-1} + \frac{\beta \phi_k}{k} \left( F\left(C-\phi_k^\top\theta-\phi_k^\top\psi_{k-1}-\phi_k^\top w_{k-1}\right) \right. \nonumber\\
	& \qquad\qquad \left. - F\left(C-\phi_k^\top\theta\right) \right) + \frac{w_{k-1}}{k}.
\end{align}
The above stochastic process is named as SPAO. With SPAO, the convergence property of the algorithm can be analyzed through that of $ w_k $.

\begin{remark}
	For general stochastic approximation methods, $ w_k $ is also used to verify the robustness of the algorithm (\!\! \cite{chen2002stochastic}, Assumption 2.7.3 and Theorem 2.7.1).
\end{remark}

To analyze the properties of SPAO $ \psi_k $, we should firstly estimate the distribution tail of $ w_k $.

\begin{lemma}\label{lemma:T}
	Let $w_k$ be defined in \eqref{defT}, and assume that
	\begin{enumerate}[label=\roman*)]
		\item $\phi_k\in \mathbb{R}^n$ is bounded;
		\item $s_k\in \{0,1\}$ is a binary random variable with expectation $F_k$, and the sequence $ \{ s_k, k\geq 1 \} $ is independent.
	\end{enumerate}
	Then, for any $\varepsilon\in(0,\frac{1}{2})$ , there exists $m>0$ such that
	\begin{equation}
		\mP\left\{\sup_{j\geq k}j^{\varepsilon}\Abs{w_j}>1\right\}=O\left(\exp(-mk^{1-2\varepsilon})\right).
	\end{equation}
\end{lemma}

The proof is given in \cref{appen:proof}.

Next, we give the following lemma to describe the distance between $ \psi_k $ and the estimation error $ \tilde{\theta}_k $ in three different senses based on \cref{lemma:T}.

\begin{lemma}\label{prop:dist}
	Assume that
	\begin{enumerate}[label=\roman*)]
		\item The system \eqref{1} with binary observation \eqref{2} satisfies \cref{a1,a2};
		\item $w_k$ is defined in \eqref{defT}, and $\psi_k = \tilde{\theta}_k - w_k$.
	\end{enumerate}
	Then, we have
	\begin{enumerate}[label={(\alph*)}]
		\item for any $ \varepsilon\in(0,\frac{1}{2}) $, there exists $ m>0 $ such that $ \mP\left\{\lVert\tilde{\theta}_k - \psi_k\rVert > k^{-\varepsilon}\right\} = O\left(\exp\left(-mk^{1-2\varepsilon}\right)\right) $;
		\item $ \lVert\tilde{\theta}_k - \psi_k\rVert = O\left(\sqrt{\ln \ln k/k}\right) $, a.s.;
		\item $ \mE\lVert\tilde{\theta}_k - \psi_k\rVert^2 = O(1/k) $.
	\end{enumerate}
\end{lemma}

\begin{proof}
	Since $ \tilde{\theta}_k - \psi_k = w_k $,  the three parts of the lemma can be obtained immediately from \cref{lemma:T}, the law of iterated logarithm (\!\!\! \cite{YSC}, Theorem 10.2.1) and $ \mE\Abs{w_k}^2 = O(1/k) $, respectively.
\end{proof}



Then, by using \cref{lemma:T,prop:dist}, the following theorem estimates the distribution tail of SPAO $ \psi_k $. 

\begin{theorem}\label{prop:in}
	Under the condition of \cref{prop:dist}, for any $\Mpri>0$ and $\varepsilon\in(0,\frac{1}{2})$, when $ k $ is sufficiently large,
	\begin{equation*}
		\left\{ \lVert\psi_k\rVert^2 < \Mpri \right\} \supseteq \left\{ \sup_{j\geq \lfloor k^{1-\varepsilon} \rfloor} j^{\varepsilon}\Abs{w_j}\leq 1 \right\}.
	\end{equation*}
	Furthermore, there exists $ m>0 $ such that
	\begin{equation*}
		\mP \left\{ \lVert\psi_k\rVert^2 \geq \Mpri \right\} = O\left( \exp\left(-mk^{(1-\varepsilon)(1-2\varepsilon)}\right) \right).
	\end{equation*}
\end{theorem}

\begin{proof}
	Set $k_s=\lfloor k^{1-\varepsilon}\rfloor$ and $k_s^\prime=k-N\lfloor \frac{k-k_s}{N} \rfloor$. It is worth mentioning that $ k_s^\prime\in [k_s,k_s+N-1] $, and $ k - k_s^\prime $ is divisible by $ N $. Assume that $\sup_{j\geq k_s} j^{\varepsilon}\Abs{w_j}\leq 1 $ is true in the rest of the proof. Then, it suffices to prove that $\Abs{\psi_k}^2<\Mpri$.
	
	We firstly simplify the recursive formula of $ \Abs{\psi_k}^2 $. By \eqref{recur_psi} and the monotonicity and Lipschitz continuity of $F(\cdot)$, for any positive real number $ b $, we have
	
	\begin{align}\label{recur:psi_pre}
		& \Abs{\psi_{k}}^2 \nonumber\\
		\leq & \Abs{\psi_{k-1}}^2 \! + \frac{2\beta\phi_k^\top\psi_{k-1}}{k} \! \left(F(C \! - \phi_k^\top\theta-\phi_k^\top w_{k-1}-\phi_k^\top\psi_{k-1})\right.\nonumber\\
		&\left.-F(C-\phi_k^\top\theta)\right)+\frac{2\psi_{k-1}^\top w_{k-1}}{k}+\frac{\left(\beta \lVert \phi_k\rVert+\lVert w_{k-1}\rVert\right)^2}{k^2}\nonumber\\
		=&\Abs{\psi_{k-1}}^2 \! + \frac{2\beta\phi_k^\top\psi_{k-1}}{k} \! \left(F(C \! - \phi_k^\top\theta-\phi_k^\top w_{k-1}-\phi_k^\top\psi_{k-1})\right.\nonumber\\
		&\qquad\qquad\left.-F(C-\phi_k^\top\theta-\phi_k^\top w_{k-1})\right)+O\left(k^{-1-\varepsilon/2}\right)\nonumber\\
		\leq&\Abs{\psi_{k-1}}^2+\frac{2\beta\phi_k^\top\psi_{k-1}}{k}\left(F(C-\phi_k^\top\theta-\phi_k^\top w_{k-1}-b)\right.\nonumber\\
		&\qquad\qquad\left.-F(C-\phi_k^\top\theta-\phi_k^\top w_{k-1})\right)I_{\{\phi_k^\top\psi_{k-1}\geq b\}}\nonumber\\
		&+\frac{2\beta\phi_k^\top\psi_{k-1}}{k}\left(F(C-\phi_k^\top\theta-\phi_k^\top w_{k-1}+b)\right.\nonumber\\
		&\qquad\qquad\left.-F(C-\phi_k^\top\theta-\phi_k^\top w_{k-1})\right)I_{\{\phi_k^\top\psi_{k-1}\leq -b\}}\nonumber\\
		&+O\left(k^{-1-\varepsilon/2}\right).
	\end{align}
	
	By \cref{a2} and the boundedness of $C-\phi_k^\top\theta-\phi_k^\top w_{k-1}$, there exists $ B>0 $ such that
	\begin{align*}
		-2\beta&\left(F(C-\phi_k^\top\theta-\phi_k^\top w_{k-1}-b)\right.\nonumber\\
		&\qquad\qquad\left.-F(C-\phi_k^\top\theta-\phi_k^\top w_{k-1})\right)>B,\\
		2\beta&\left(F(C-\phi_k^\top\theta-\phi_k^\top w_{k-1}+b)\right.\nonumber\\
		&\qquad\qquad\left.-F(C-\phi_k^\top\theta-\phi_k^\top w_{k-1})\right)>B,
	\end{align*}
	which together with \eqref{recur:psi_pre} implies
	\begin{equation*}
		\Abs{\psi_{k}}^2 \! \leq \Abs{\psi_{k-1}}^2-\frac{B\lvert \phi_k^\top\psi_{k-1}\rvert}{k}I_{ \{\lvert \phi_k^\top\psi_{k-1}\rvert\geq b \} }+O\!\left(k^{-1-\varepsilon/2}\right)\!.
	\end{equation*}
	
	Set $ b=\frac{\sqrt{\delta\Mpri}}{2} $. Then, we have 
	\begin{align}\label{ineq:psi^2}
		& \Abs{\psi_{k+N}}^2 \nonumber\\
		\leq& \Abs{\psi_{k}}^2 - \sum_{i=k+1}^{k+N} \frac{B\lvert \phi_i^\top\psi_{i-1}\rvert}{i}I_{ \left\{\lvert \phi_i^\top\psi_{i-1}\rvert\geq \frac{\sqrt{\delta\Mpri}}{2} \right\} } \nonumber\\
		& + \sum_{i=k+1}^{k+N}O\left(i^{-1-\varepsilon/2}\right)
	\end{align}
{
	By \eqref{defT}, \eqref{recur_psi} and \cref{a1},  
	\begin{align}\label{ineq:psi-psi}
		& \Abs{\psi_k-\psi_{k-1}} 
		\leq \frac{1}{k}\left( 2\beta \Abs{\phi_k} + \Abs{w_{k-1}} \right) \nonumber\\
		\leq & \frac{1}{k}\left( 2 \beta M + 2 \beta M \right)
		= \frac{4\beta M}{k}. 
	\end{align}
	Note that $ \Mpri > 0 $. Then, by \cref{lemma:Connection}, when $ k $ is sufficiently large, there exists $k^\prime \in [k+1,k+N]$ such that
	\begin{align*}
		\lvert \phi_{k^\prime}^\top\psi_{k^\prime-1}\rvert 
		\geq \sqrt{\frac{\delta}{2}} \Abs{\psi_k} I_{\left\{\Abs{\psi_k}^2 \geq \frac{\Mpri}{2}\right\}},
	\end{align*}
	which implies
	\begin{align}\label{seteq:sup}
		& \left\{ \lvert \phi_{k^\prime}^\top\psi_{k^\prime-1}\rvert\geq \frac{\sqrt{\delta\Mpri}}{2} \right\} \nonumber\\
		\supseteq & \left\{ \sqrt{\frac{\delta}{2}} \Abs{\psi_k} \geq \frac{\sqrt{\delta\Mpri}}{2} \right\} \cap \left\{ \Abs{\psi_k}^2 \geq \frac{\Mpri}{2} \right\} \nonumber\\
		\supseteq & \left\{ \Abs{\psi_k}^2 \geq \frac{\Mpri}{2} \right\}. 
	\end{align}
	Then, by \eqref{ineq:psi^2} and \eqref{seteq:sup},
	\begin{align}
		& \Abs{\psi_{k+N}}^2 \nonumber\\
		\leq& \Abs{\psi_{k}}^2 - \frac{B\lvert \phi_{k^\prime}^\top\psi_{k^\prime-1}\rvert}{k+N}I_{ \left\{\lvert \phi_{k^\prime}^\top\psi_{k^\prime-1}\rvert\geq \frac{\sqrt{\delta\Mpri}}{2} \right\} }\!\!+ O\left(k^{-1-\varepsilon/2}\right) \nonumber\\
		\leq& \Abs{\psi_{k}}^2 - \frac{B\sqrt{\delta}}{\sqrt{2}}\frac{\Abs{\psi_k}}{k+N}I_{ \left\{\Abs{\psi_{k}}^2 \geq \frac{\Mpri}{2}\right\} }+ O\left(k^{-1-\varepsilon/2}\right).\!\!\! 
	\end{align}
}

	Hence, when $ k=k_s^\prime + N(t-1) $, we have
	\begin{align}\label{recur:psi_simpl}
		& \Abs{\psi_{k_s^\prime+Nt}}^2 \nonumber\\
		\leq& \Abs{\psi_{k_s^\prime+N(t-1)}}^2 
		- \frac{B\sqrt{\delta}}{\sqrt{2}}\frac{\Abs{\psi_{k_s^\prime+N(t-1)}}}{k_s^\prime+Nt}I_{ \left\{\lVert\psi_{k_s^\prime+Nt}\rVert^2 \geq \frac{\Mpri}{2} \right\} }\nonumber\\
		&+O\left((k_s^\prime + Nt)^{-1-\varepsilon/2}\right).
	\end{align}
	\vspace{1pt}
	
	Since $ \lim\limits_{k\to\infty}k_s^\prime = \infty $, we have
	\begin{equation*}
		\lim\limits_{k\to\infty}\sum_{t=1}^{\infty}(k_s^\prime + Nt)^{-1-\varepsilon/2} = 0.
	\end{equation*}
	 Then, by \eqref{recur:psi_simpl} and \cref{lemma:Seq} in \cref{appen:proof}, when $ k $ is sufficiently large, we have
	\begin{equation*}
		\Abs{\psi_{k}}^2
		= \Abs{\psi_{k_s^\prime+N\lfloor \frac{k-\lfloor k^{1-\varepsilon}\rfloor}{N} \rfloor}}^2
		< \max\left\{ \Mpri , \Delta_k^2 +\frac{\Mpri}{2} \right\},
	\end{equation*}
	where
	\begin{equation}
		\Delta_k =\left( \Abs{\psi_{k_s^\prime}} - \frac{B\sqrt{\delta}}{2\sqrt{2}N}\ln \left( \frac{\lfloor \frac{k-\lfloor k^{1-\varepsilon}\rfloor}{N} \rfloor + \frac{k_s^\prime}{N}+ 1}{\frac{k_s^\prime}{N}+1} \right) \right)^+. 
	\end{equation}
	Note that $ \ln \left( \frac{\lfloor (k-\lfloor k^{1-\varepsilon}\rfloor )/N \rfloor+k_s^\prime/N+1}{k_s^\prime/N+1} \right)  $ is of the same order as $ \ln k $ since $ k_s^\prime = O (k^{1-\varepsilon}) $. 
	And, by \cref{coro:worst_SPAO} in \cref{appen:proof}, it holds that $ \psi_{k_s^\prime} = O\left(\sqrt{\ln k}\right) $. Then, when $ k $ is sufficiently large, $ \Delta_k = 0 $ and $ \Abs{\psi_k}^2 < \Mpri $, which proves the theorem.
\end{proof}

\begin{remark}
	 The distribution tail estimation of SPAO $ \psi_k $ in \cref{prop:in} can be promoted to \cref{prop:prom_psi_tail} in \cref{appen:proof}. 
\end{remark}

\begin{remark}
	It is worth noting that the constructed SPAO can not only be adapted to Algorithm \eqref{algo}, but also can be extended to a class of identification algorithms of the {binary-valued} systems. The details are given in \cref{appen:extend}. 
\end{remark}

\subsection{Estimate of the distribution tail}\label{subsec:DistTailEstm}

In this subsection, the distribution tail of the estimation error will be estimated.

\begin{theorem}\label{thm:tail}
	If the system \eqref{1} with binary observations \eqref{2} satisfies \cref{a1,a2}, then for any $\Mpri>0$ and $\varepsilon>0$, there exists $m>0$ such that
	\begin{equation*}
		\mP\left\{ \sup\limits_{j\geq k}\lVert\tilde{\theta}_j\rVert^2 \geq \Mpri \right\}=O\left(\exp(-mk^{1-\varepsilon})\right).
	\end{equation*}
\end{theorem}

\begin{proof}
	Reminding that $\tilde{\theta}_k=\psi_k+w_k$, by \cref{prop:in}, for sufficiently large $ k $, we have
	\begin{align*}
		& \left\{ \sup_{j\geq \lfloor k^{1-\varepsilon} \rfloor} j^{\varepsilon}\Abs{w_j}\leq 1 \right\} \nonumber\\
		\subseteq & \left\{ \lVert\psi_k\rVert^2<\frac{\Mpri}{4} \right\}\cap\left\{ \lVert w_k\rVert^2\leq\frac{\Mpri}{4} \right\}
		\subseteq \left\{ \lVert \tilde{\theta}_k \rVert^2< \Mpri \right\},
	\end{align*}
	and hence,
	\begin{align*}
		\left\{ \sup\limits_{j\geq k}\lVert\tilde{\theta}_j\rVert^2 \geq \Mpri \right\}
		\subseteq & \bigcup_{j\geq k}\left\{ \sup_{j_0\geq \lfloor j^{1-\varepsilon} \rfloor} j_0^{\varepsilon}\Abs{w_{j_0}}> 1 \right\} \nonumber\\
		= & \left\{ \sup_{j\geq \lfloor k^{1-\varepsilon} \rfloor} j^{\varepsilon}\Abs{w_j}> 1 \right\}.
	\end{align*}
	So, by \cref{lemma:T}, 
	\begin{equation*}
		\mP\{ \sup_{j\geq k}\lVert\tilde{\theta}_j\rVert^2 \geq \Mpri \}=O\left(\exp(-mk^{(1-\varepsilon)(1-2\varepsilon)})\right). 
	\end{equation*}
	Thus, the theorem can be proved by the arbitrariness of $\varepsilon$.
\end{proof}

\begin{remark}
	\cref{thm:tail} estimates the distribution tail of the estimation error $ \tilde{\theta}_k $. For the convergence analysis of identification algorithms, the existing works are usually interested in the asymptotic properties of the estimation error distribution. For example, the asymptotic normality of $ \rho_k^{-1/2}\tilde{\theta}_k $ is given for general stochastic approximation algorithms under different conditions (\!\!\cite{chen2002stochastic}, Section 3.3 and \cite{fabian1968asymptotic}). For the finite-valued system with i.i.d. inputs and designable quantizers, \cite{You2015recursive} also analyzes the asymptotic normality of the algorithm. Compared with the asymptotic normality, \cref{thm:tail} weakens the description of the estimate distribution in the neighborhood of $ \theta $, but gives a better description on the exponential tail of the estimation error. This helps to obtain the almost sure and mean square convergence of the algorithm.
\end{remark}

\begin{theorem}\label{thm:conv}
	Under the condition of \cref{thm:tail}, Algorithm \eqref{algo} converges to $\theta$ in both almost sure and mean square sense.
\end{theorem}

\begin{proof}
	The almost sure convergence can be immediately obtained by \cref{thm:tail}.
	
	By \cref{thm:tail,coro:worst_EstErr} in \cref{appen:proof}, for any $ \Mpri>0 $ and $\varepsilon>0$, there exists $m>0$ such that
	\begin{align*}
		\!\!\!\mE \lVert \tilde{\theta}_k \rVert^2
		=& \int_{\{\lVert \tilde{\theta}_k^2 \rVert < \Mpri\}}\lVert \tilde{\theta}_k \rVert^2 \td\mP + \int_{\{\lVert \tilde{\theta}_k^2 \rVert \geq \Mpri\}}\lVert \tilde{\theta}_k \rVert^2 \td\mP \nonumber\\
		<& \Mpri + O\left( \ln k \cdot \exp(-mk^{1-\varepsilon}) \right) = \Mpri + o(1). 
	\end{align*}
	Thus, the mean square convergence can be obtained by the arbitrariness of $ \Mpri $.
\end{proof}

\begin{remark}
	When the inputs are periodic, the mean square convergence of the empirical measurement method without truncation is also proved by the estimation of the distribution tail \cite{zhao2018consensus}.
	The distribution tail of the estimate is relatively easy to be obtained for the empirical measurement method, because there is a direct connection between the average of the {binary-valued} observations and the distribution tail. But, in the SA-based algorithm, the relationship is much more complicated. Therefore, SPAO is constructed to reveal the connection. 
\end{remark}


\subsection{Almost sure convergence rate}\label{subsec:ASConvRate}

This subsection will estimate the almost sure convergence rate of the SA-based algorithm.

Before the analysis, we define
\begin{equation}\label{def:udf_func}
	\udf(x) = 
	\sup\limits_{z > M\Abs{\theta}+x} \inf\limits_{t\in [C-z,C+z] } f(t)>0,\quad \forall x\geq 0,
\end{equation}
and
\begin{equation}\label{def:udf}
	\udf = \udf(0).
\end{equation}
The convergence rate of the algorithm depends on $ \udf $.

\begin{remark}\label{remark:udf}
	Under \cref{a2}, $ \udf $ is the lower bound of $ f(C-\phi_k^\top\theta) $ for all possible regressors $ \phi_k $. The following lemma gives properties of $ \udf(\cdot) $ and $ \udf $.	
\end{remark}

{\begin{lemma}\label{lemma:udf}
	Under \cref{a2}, $ \udf(\cdot) $ and $ \udf $ have the following properties. 
	\begin{enumerate}[label={(\alph*)}]
		\item $\udf(\cdot)$ is non-increasing and right continuous;
		\item $ \lim\limits_{x\to 0}\udf(x) = \udf $; 
		\item $ \udf(x) \leq
		\inf_{t\in [C-M\Abs{\theta}-x,C+M\Abs{\theta}+x] } f(t) $; 
		\item If, in addition, $f(\cdot)$ is locally Lipschitz continuous, then so is $\udf(\cdot)$.
	\end{enumerate}
\end{lemma}

The proof is given in \cref{appen:proof}. }

The almost sure convergence rate of the algorithm can be achieved through that of $ \psi_k $.

\begin{theorem}\label{prop:psi_as_rate}
	Under the condition of \cref{prop:dist}, for any $ \varepsilon>0 $, we have
	\begin{equation}\label{rate:as}
		\psi_k =
		\begin{cases}
			O\left(\sqrt{\frac{\ln\ln k}{k}}\right), & \eta>\frac{1}{2};\\
			O\left(\frac{1}{k^{\eta-\varepsilon}}\right), & \eta\leq\frac{1}{2},
		\end{cases}\quad \text{a.s.,}
	\end{equation}
	where $ \eta = \beta\delta\udf $ with $ \udf $ defined in \eqref{def:udf}. If $ f(\cdot) $ is assumed to be locally Lipschitz continuous, then the almost sure convergence rate can be promoted into
	\begin{equation}\label{rate:as_promote}
		\psi_k =
		\begin{cases}
			O\left(\sqrt{\frac{\ln\ln k}{k}}\right), & \eta>\frac{1}{2};\\
			O\left(\ln k\sqrt{\frac{\ln\ln k}{k}}\right), & \eta=\frac{1}{2};\\
			O\left(\frac{1}{k^{\eta}}\right), & \eta<\frac{1}{2},
		\end{cases}\quad \text{a.s.} 
	\end{equation}
\end{theorem}

\begin{proof}
	The proof is based on \cref{lemma:gen_zhao18} in \cref{appen:proof}.
	
	For the proof of \eqref{rate:as}, we firstly simplify the recursive formula of $ \Abs{\psi_k} $. {By \cref{a2}, $ F(\cdot) $ is Lipschitz continuous, which implies $ \sup_{x\in\mathbb{R}} f(x) < \infty $.}
	Then, in \eqref{recur_psi}, by Lagrange mean value theorem (\!\!\! \cite{zorich}, Theorem 5.3.1), there exist $ \xi_k $ between $ C-\phi_k^\top\theta-\phi_k^\top\psi_{k-1} $ and $ C-\phi_k^\top\theta $, {and $ \xi_k^\prime $ between $ C-\phi_k^\top\theta-\phi_k^\top w_{k-1}-\phi_k^\top\psi_{k-1} $ and $ C-\phi_k^\top\theta-\phi_k^\top\psi_{k-1} $} such that
	{\begin{align}\label{eq:Lagrange}
		&F(C-\phi_k^\top\theta-\phi_k^\top w_{k-1}-\phi_k^\top\psi_{k-1}) - F(C-\phi_k^\top\theta)\nonumber\\
		= & F(C-\phi_k^\top\theta-\phi_k^\top w_{k-1}-\phi_k^\top\psi_{k-1}) \nonumber\\
		& - F(C-\phi_k^\top\theta-\phi_k^\top\psi_{k-1}) \nonumber\\
		& + F(C-\phi_k^\top\theta-\phi_k^\top\psi_{k-1}) - F(C-\phi_k^\top\theta)\nonumber\\
		= & f(\xi_k^\prime) \phi_k^\top w_{k-1} + f(\xi_k)\phi_k^\top\psi_{k-1} \nonumber\\
		=& f(\xi_k)\phi_k^\top\psi_{k-1}+O\left(w_{k-1}\right).
	\end{align}}
	Then, by the law of iterated logarithm (\!\cite{YSC}, Theorem 10.2.1),
	\begin{align*}
		& F(C-\phi_k^\top\theta-\phi_k^\top w_{k-1}-\phi_k^\top\psi_{k-1}) - F_k \nonumber\\
		= & f(\xi_k)\phi_k^\top\psi_{k-1}+O\left(\sqrt{\frac{\ln \ln k}{k}}\right),\ \text{a.s.},
	\end{align*}
	which together with \eqref{recur_psi} implies
	\begin{equation*}
		\psi_k
		= \left( I_n - \frac{\beta f(\xi_k)}{k} \phi_k \phi_k^\top \right) \psi_{k-1} +O\left(\sqrt{\frac{\ln \ln k}{k^3}}\right), \quad \text{a.s.}
	\end{equation*}
	Note that except for the first few steps, we have
	\begin{align}\label{ineq:matrix_prod}
		&\Abs{\prod_{i=k-N+1}^k\left(I_n - \frac{\beta f(\xi_i)}{i} \phi_i \phi_i^\top \right)}\nonumber\\
		\leq& \Abs{I_n - \frac{\beta}{k}\sum_{i=k-N+1}^k f(\xi_i)\phi_i \phi_i^\top } + O\left(\frac{1}{k^2}\right)\nonumber\\
		\leq & \Abs{I_n - \frac{\beta}{k}\udf\left( \max\limits_{k-N<i\leq k} \! \abs{\phi_i^\top\psi_{i-1}} \right) \! \sum_{i=k-N+1}^k \!\!\! \phi_i \phi_i^\top } + O\left(\frac{1}{k^2}\right)\nonumber\\
		\leq& \left( 1 - \frac{\beta \delta N}{k} \udf\left( \max\limits_{k-N<i\leq k} \abs{\phi_i^\top\psi_{i-1}} \right) \right) + O\left(\frac{1}{k^2}\right),
	\end{align}
	where $ \udf(\cdot) $ is defined in \eqref{def:udf_func}. Denote
	\begin{equation*}\label{def:udf_k|N}
		\udf_{k|N} = \udf\left( \max\limits_{k-N<i\leq k} \abs{\phi_i^\top\psi_{i-1}} \right).
	\end{equation*}
	Then, we have
	\begin{equation}\label{recur:psi_abs_model}
		\Abs{\psi_k}
		\leq  \left( 1 - \frac{\beta \delta N}{k}\udf_{k|N} \right) \Abs{\psi_{k-N}} + O\left(\sqrt{\frac{\ln \ln k}{k^3}}\right), \quad \text{a.s.}
	\end{equation}
	
	By (b) of \cref{lemma:udf}, $ \lim\limits_{k\to\infty}\udf_{k|N} = \udf  $ almost surely. 
	{Then, there exists $ \epp \in (0, \min\{\eta-\frac{1}{2},\varepsilon\}) $ if $ \eta > \frac{1}{2} $, and $ \epp \in (0,\min\{\eta,\varepsilon\}) $ otherwise.}
	Therefore,  there almost surely exists $ k_a $ such that for all $ k\geq k_a $, we have $ \beta\delta\udf_{k|N}>\beta\delta\udf-\varepsilon_1 = \eta - \epp $, and thus by \eqref{recur:psi_abs_model},
	\begin{equation*}
		\Abs{\psi_k}
		\leq  \left( 1 - \frac{N}{k}(\eta-\epp)  \right) \Abs{\psi_{k-N}} + O\left(\sqrt{\frac{\ln \ln k}{k^3}}\right),\quad \text{a.s.}
	\end{equation*}
	
	{If} $ k-k_a $ is divisible by $ N $, then by \cref{lemma:gen_zhao18} in \cref{appen:proof}, one can get
	\begin{align*}
		& \Abs{\psi_k}\nonumber\\
		\leq& \prod_{i=1}^{\frac{k-k_a}{N}}\left(1-\frac{N(\eta-\epp)}{k_a+iN}\right)\Abs{\psi_{k_a}}\nonumber\\
		&+ O\left(\sum_{l=1}^{\frac{k-k_a}{N}} \prod_{i=l+1}^{\frac{k-k_a}{N}} \left(1-\frac{N(\eta-\epp)}{k_a+iN}\right) \sqrt{\frac{\ln \ln (k_a+lN)}{(k_a+lN)^3}} \right)\nonumber\\
		=& O\left(\frac{1}{k^{\eta-\epp}}\right) \nonumber\\ 
		& + O\left( \sum_{l=1}^{\frac{k-k_a}{N}} \prod_{i=l+1}^{\frac{k-k_a}{N}} \left(1-\frac{N(\eta-\epp)}{k_a+iN}\right) \sqrt{\frac{\ln \ln (l+2)}{l^3}}\right)\nonumber\\
		=& \begin{cases}
			O\left(\sqrt{\frac{\ln\ln k}{k}}\right), & \eta-\epp>\frac{1}{2};\\
			O\left(\ln k\sqrt{\frac{\ln\ln k}{k}}\right), & \eta-\epp=\frac{1}{2};\\
			O\left(\frac{1}{k^{\eta-\epp}}\right), & \eta-\epp<\frac{1}{2},
		\end{cases}\quad \text{a.s.}
	\end{align*}
	{By the settings of $ \epp $, we have $ \epp \leq \varepsilon $, and $ \eta - \epp > \frac{1}{2} $ if and only if $ \eta > \frac{1}{2} $. Therefore, 
	\begin{align}\label{ineq:rate_psi_divisible}
		\Abs{\psi_k} = 
		\begin{cases}
			O\left(\sqrt{\frac{\ln\ln k}{k}}\right), & \eta>\frac{1}{2};\\
			O\left(\frac{1}{k^{\eta-\varepsilon}}\right), & \eta\leq\frac{1}{2}, 
		\end{cases}\quad \text{a.s.}
	\end{align}}

	{If $ k-k_a $ is not divisible by $ N $, then there exists an integer $ \kappa \in [k-N+1,k] $ such that $ \kappa - k_a $ is divisible by $ N $. By \eqref{ineq:psi-psi}, $ \Abs{\psi_k - \psi_\kappa} \leq \sum_{i=\kappa+1}^{k} \Abs{\psi_i - \psi_{i-1}} \leq \frac{4\beta M (k-\kappa-1)}{k} \leq \frac{4\beta M N}{k} $. Hence, by \eqref{ineq:rate_psi_divisible}, 
	\begin{align*}
		& \Abs{\psi_k} 
		= \Abs{\psi_k - \psi_\kappa}+\Abs{\psi_\kappa} \nonumber\\
		= &\begin{cases}
			O\left(\sqrt{\frac{\ln\ln k}{k}}\right) + O\left( \frac{1}{k} \right), & \eta>\frac{1}{2};\\
			O\left(\frac{1}{k^{\eta-\varepsilon}}\right) + O\left( \frac{1}{k} \right), & \eta\leq\frac{1}{2}, 
		\end{cases}\quad \text{a.s.} \nonumber\\
		= & \begin{cases}
			O\left(\sqrt{\frac{\ln\ln k}{k}}\right), & \eta>\frac{1}{2};\\
			O\left(\frac{1}{k^{\eta-\varepsilon}}\right), & \eta\leq\frac{1}{2}, 
		\end{cases}\quad \text{a.s.}
	\end{align*}
} 
	\eqref{rate:as} is thereby proved.
	
	Then, we now prove \eqref{rate:as_promote}. For sufficiently large $ k $,
	\begin{equation}\label{ineq:1-udf}
		1-\frac{\beta \delta N}{k}\udf_{k|N}\leq \left( 1 + \frac{\beta \delta N}{k}\left(\udf- \udf_{k|N}\right)\right)\left(1-\frac{\beta \delta N}{k}\udf\right).
	\end{equation}
	which together with \eqref{recur:psi_abs_model} implies
	\begin{align}\label{recur:psi/prod}
		&\left.\Abs{\psi_k}\middle/\prod_{i=N+1}^k \left( 1 + \frac{\beta\delta N}{i}\left(\udf- \udf_{i|N}\right)\right)\right. \nonumber\\
		\leq & \left(1-\frac{N\eta}{k}\right) \left.\Abs{\psi_{k-N}}\middle/\prod_{i=N+1}^{k-N} \left( 1 + \frac{\beta\delta N}{i}\left(\udf- \udf_{i|N}\right)\right)\right. \nonumber\\
		& + O\left(\sqrt{\frac{\ln \ln k}{k^3}}\right),\quad \text{a.s.}
	\end{align}
	By \eqref{rate:as} and (d) of \cref{lemma:udf}, $ \udf- \udf_{i|N} $ converges to $ 0 $ at a polynomial rate. Hence, we have 
	\begin{equation*}
		\prod_{i=N+1}^\infty \left( 1 + \frac{\beta\delta N}{i}\left(\udf- \udf_{i|N}\right)\right)<\infty.
	\end{equation*}
	Then, \eqref{rate:as_promote} can be proved by \eqref{recur:psi/prod} and \cref{lemma:gen_zhao18}.
\end{proof}


Then, the almost sure convergence rate of the algorithm can be obtained by \cref{prop:psi_as_rate}. 

\begin{theorem}\label{thm:as_rate}
	Under the condition of \cref{thm:tail}, for any $ \varepsilon>0 $,
	\begin{equation*}
		\tilde{\theta}_k =
		\begin{cases}
			O\left(\sqrt{\frac{\ln\ln k}{k}}\right), & \eta>\frac{1}{2};\\
			O\left(\frac{1}{k^{\eta-\varepsilon}}\right), & \eta\leq\frac{1}{2},
		\end{cases}\quad \text{a.s,}
	\end{equation*}
	where $ \eta = \beta\delta\udf $ with $ \udf $ defined in \eqref{def:udf}. If the density function $ f(\cdot) $ is assumed to be locally Lipschitz continuous, then the almost sure convergence rate can be promoted into
	\begin{equation*}
		\tilde{\theta}_k =
		\begin{cases}
			O\left(\sqrt{\frac{\ln\ln k}{k}}\right), & \eta>\frac{1}{2};\\
			O\left(\ln k\sqrt{\frac{\ln\ln k}{k}}\right), & \eta=\frac{1}{2};\\
			O\left(\frac{1}{k^{\eta}}\right), & \eta<\frac{1}{2},
		\end{cases}\quad \text{a.s.} 
	\end{equation*}
\end{theorem}

\begin{proof}
	The theorem can be obtained by \cref{prop:dist,prop:psi_as_rate}.
\end{proof}

\begin{remark}
	By \cref{thm:as_rate}, the algorithm may not achieve the optimal almost sure convergence rate when the coefficient $ \eta $ is smaller than $ 1/2 $. Since $ \eta = \beta\delta\udf $, the convergence rate of the algorithm depends on the step-size, the inputs, the noise distribution and the relationship between the threshold $ C $ and $ M\Abs{\theta} $. However, $ M\Abs{\theta} $ relies on the true parameter $ \theta $. Thus, the almost sure convergence rate of Algorithm \eqref{algo} cannot be known without \emph{a priori} information on $ \theta $. The problem can be solved if the step-size is designed as $ \rho_k = \beta_k/k $, where
	\begin{equation}
		\beta_k > \left. 1 \middle/ \left( 2\delta\sup\limits_{z>M\lVert\hat{\theta}_k\rVert} \inf\limits_{t\in [C-z,C+z] } f(t)  \right) \right. .
	\end{equation}
	The analysis for the modified algorithm is similar to the algorithm with time-invariant $ \beta $.
\end{remark}

\begin{remark}
	For the identification problem of stochastic finite-valued systems, $ O(\sqrt{\ln \ln k/k}) $ is the best  almost sure convergence rate. In the periodic input case, the empirical measurement algorithm in \cite{Wang2003System} generates a maximum likelihood estimate (\!\! \cite{godoy2011on}, Lemma 4). The almost sure convergence rate of the empirical measurement algorithm is $ O(\sqrt{\ln \ln k/k}) $ \cite{Mei2014almost}. In the non-periodic input case, \cref{thm:as_rate} appears to be the first to achieve the almost sure convergence rate of $ O(\sqrt{\ln \ln k /k }) $ theoretically. \cite{guo2013recursive} achieves the almost sure convergence rate of $ O(\sqrt{\ln k/k}) $ for the recursive projection method. And, the almost sure convergence rate of stochastic approximation algorithms with expanding truncations is $ O(1/k^{\varepsilon}) $ for $ \varepsilon\in(0,1/2) $ \cite{Song2018recursive}. When properly selecting $ \beta $, the almost sure convergence rate of Algorithm \eqref{algo} is better than both of them.
\end{remark}


\subsection{Mean square convergence rate}\label{subsec:M2ConvRate}

This subsection will estimate the mean square convergence rate of the SA-based algorithm.

\begin{theorem}\label{thm:m2_rate}
	Under the condition of \cref{thm:tail}, for any $ \varepsilon>0 $,
	\begin{equation}\label{rate:the_ms}
		\mE\lVert\tilde{\theta}_k\rVert^2 =
		\begin{cases}
			O\left(\frac{1}{k}\right), & \eta>\frac{1}{2};\\
			O\left(\frac{1}{k^{2\eta-\varepsilon}}\right), & \eta\leq\frac{1}{2},
		\end{cases}
	\end{equation}
	where $ \eta = \beta\delta\udf $ with $ \udf $ defined in \eqref{def:udf}. If $ f(\cdot) $ is assumed to be locally Lipschitz continuous, then the mean square convergence rate can be promoted into
	\begin{equation}\label{rate:the_ms_prom}
		\mE\lVert\tilde{\theta}_k\rVert^2 =
		\begin{cases}
			O\left(\frac{1}{k}\right), & \eta>\frac{1}{2};\\
			O\left(\frac{\ln k}{k}\right), & \eta=\frac{1}{2};\\
			O\left(\frac{1}{k^{2\eta}}\right), & \eta<\frac{1}{2}.
		\end{cases}
	\end{equation}
\end{theorem}

\begin{proof}
	To prove \eqref{rate:the_ms}, we firstly simplify the recursive formula of $ \mE\lVert\tilde{\theta}_k\rVert^2 $.
	
	By \eqref{algo} and the Lagrange mean value theorem (\!\! \cite{zorich}, Theo-rem 5.3.1), there exists $ \zeta_k $ between $ C-\phi_k^\top\theta $ and $ C-\phi_k^\top\theta - \phi_k^\top\tilde{\theta}_{k-1} $ such that
	\begin{align*}
		\tilde{\theta}_k
		=& \tilde{\theta}_{k-1} + \frac{\beta \phi_k}{k}\left( \hat{F}_k - F_k \right) + \frac{\beta \phi_k}{k}\left( F_k - s_k \right)\nonumber\\
		=& \left( I_n - \frac{\beta}{k}f(\zeta_k)\phi_k\phi_k^\top \right)\tilde{\theta}_{k-1} + \frac{\beta \phi_k}{k}\left( F_k - s_k \right)\nonumber\\
		=& \prod_{i=k-N+1}^k \left( I_n - \frac{\beta}{i}f(\zeta_i)\phi_i\phi_i^\top \right)\tilde{\theta}_{k-N}\nonumber\\
		&+ \sum_{l=k-N+1}^{k} \prod_{i=l+1}^{k}\left( I_n - \frac{\beta}{i}f(\zeta_i)\phi_i\phi_i^\top \right)\frac{\beta \phi_l}{l}\left( F_l - s_l \right)\nonumber\\
		=& \prod_{i=k-N+1}^k \left( I_n - \frac{\beta}{i}f(\zeta_i)\phi_i\phi_i^\top \right)\tilde{\theta}_{k-N} \nonumber\\
		& + \sum_{l=k-N+1}^{k} \frac{\beta \phi_l}{l}\left( F_l - s_l \right) + O\left(\frac{1}{k^2}\right).
	\end{align*}
	Similar to \eqref{ineq:matrix_prod}, except for the first few steps, we have
	\begin{align*}
		& \Abs{\prod_{i=k-N+1}^k\left( I_n - \frac{\beta f(\zeta_i)}{i}\phi_i\phi_i^\top \right)} \\
		\leq & \left( 1 -  \frac{\beta \delta N}{k}\udf_{k|N}^\prime \right) + O\left(\frac{1}{k^2}\right),
	\end{align*}
	where $ \udf_{k|N}^\prime =  \udf\left( \max\limits_{k-N<i\leq k} \abs{\phi_i^\top\tilde{\theta}_{i-1}} \right) $, and $ \udf(\cdot) $ is defined in \eqref{def:udf_func}. Besides, noticing that $ \tilde{\theta}_{k-N} $ is independent of $ \sum_{l=N+1}^{k} \frac{\beta \phi_l}{l}\left( F_l - s_l \right) $, we have
	\begin{align*}
		&\mE \left[ \left( \sum_{l=k-N+1}^{k} \frac{\beta \phi_l}{l}\left( F_l - s_l \right) \right)^\top \right. \nonumber\\
		& \qquad \left. \cdot \prod_{i=k-N+1}^k \left( I_n - \frac{\beta}{i}f(\zeta_i)\phi_i\phi_i^\top \right)\tilde{\theta}_{k-N} \right] \nonumber\\
		=& \mE \left[ \left( \sum_{l=k-N+1}^{k} \frac{\beta \phi_l}{l}\left( F_l - s_l \right) \right)^\top  \right. \nonumber\\
		& \qquad \left. \cdot \left(\prod_{i=k-N+1}^k \left( I_n - \frac{\beta}{i}f(\zeta_i)\phi_i\phi_i^\top \right) - I_n \right)\tilde{\theta}_{k-N} \right] \nonumber\\
		=& O\left(\frac{1}{k^2}\right).
	\end{align*}
	Therefore, for sufficiently large $ k $, one can get
	\begin{equation}\label{recur:mE_pre}
		\mE\lVert\tilde{\theta}_k\rVert^2
		\leq \mE \left[\left( 1 - \frac{\beta \delta N}{k} \udf_{k|N}^\prime \right)^2 \lVert\tilde{\theta}_{k-N}\rVert^2 \right] + O\left(\frac{1}{k^2}\right).
	\end{equation}
	
	By \cref{thm:tail} and (b) of \cref{lemma:udf}, $ \mP\{ \udf_{k|N}^\prime < \udf - \frac{\varepsilon}{2\beta \delta} \} = O(\exp(-mk^{1/2})) $. Hence, by \cref{coro:worst_EstErr} in \cref{appen:proof}, we have
	\begin{align*}
		&\mE \left[\left( 1 - \frac{\beta \delta N}{k} \udf_{k|N}^\prime  \right)^2 \lVert\tilde{\theta}_{k-N}\rVert^2 \right] \nonumber\\
		\leq& \int_{\left\{ \udf_{k|N}^\prime \geq \udf - \frac{\varepsilon}{2\beta \delta} \right\}} \left( 1 - \frac{N}{k} \left( \eta - \frac{\varepsilon}{2} \right) \right)^2 \lVert\tilde{\theta}_{k-N}\rVert^2 \td\mP \nonumber\\
		& + \int_{\left\{ \udf_{k|N}^\prime < \udf - \frac{\varepsilon}{2\beta \delta} \right\}} \lVert\tilde{\theta}_{k-N}\rVert^2 \td\mP \nonumber\\
		=&\left( 1 - \frac{N}{k} \left( \eta - \frac{\varepsilon}{2} \right) \right)^2 \mE\lVert\tilde{\theta}_{k-N}\rVert^2 \nonumber\\
		& + O\left(\ln k\cdot\exp(-mk^{1/2})\right).
	\end{align*}
	Substituting the above estimate into \eqref{recur:mE_pre} gives
	\begin{equation*}
		\mE\lVert\tilde{\theta}_k\rVert^2
		\leq \left( 1 - \frac{N}{k} \left( \eta - \frac{\varepsilon}{2} \right) \right)^2 \mE\lVert\tilde{\theta}_{k-N}\rVert^2 + O\left(\frac{1}{k^2}\right).
	\end{equation*}
	Thus, \eqref{rate:the_ms} can be proved by \cref{lemma:gen_zhao18} in \cref{appen:proof}.
	
	Then, we prove \eqref{rate:the_ms_prom}. Similar to \eqref{ineq:1-udf}, for sufficiently large $ k $,
	\begin{equation*}
		1-\frac{\beta \delta N}{k}\udf_{k|N}^\prime\leq \left( 1 + \frac{\beta \delta N}{k}\left(\udf- \udf_{k|N}^\prime\right)\right)\left(1-\frac{\beta \delta N}{k}\udf\right).
	\end{equation*}
	Therefore, by \eqref{recur:mE_pre} and $ \eta = \beta\delta\udf $, one can get
	\begin{align}\label{recur:m2prom_pre}
		\mE\lVert\tilde{\theta}_k\rVert^2
		\leq& \left(1-\frac{N\eta}{k}\right)^2 \mE \left[\left( 1 + \frac{\beta \delta N}{k}\left(\udf- \udf_{k|N}^\prime\right)\right)^2  \right.  \nonumber\\
		& \qquad\qquad \left. \cdot \lVert\tilde{\theta}_{k-N}\rVert^2 \right] + O\left(\frac{1}{k^2}\right).
	\end{align}
	
	By (d) of \cref{lemma:udf}, since $ f(\cdot) $ is assumed to be locally Lipschitz continuous here, $ \udf(\cdot) $ is also locally Lipschitz continuous. Hence, if $ \lVert \tilde{\theta}_j \rVert \leq j^{-\varepsilon^\prime} $ for $ \varepsilon^\prime>0 $ and all $ j = k-N+1, \ldots , k $, then there exists $ L>0 $ such that $ \udf -\udf_{k|N}^\prime\leq L k^{-\varepsilon^\prime}$, which together with \cref{coro:worst_EstErr,coro:tail_prom} in \cref{appen:proof} implies that there exist positive numbers $ m $ and $ \varepsilon $ such that
	\begin{align*}
		&\mE \left[\left( 1 + \frac{\beta \delta N}{k}\left(\udf- \udf_{k|N}^\prime\right)\right)^2 \lVert\tilde{\theta}_{k-N}\rVert^2 \right] \nonumber\\
		\leq& \left(1+\frac{\beta\delta N L}{k^{1+\varepsilon^\prime}}\right)^2 \int_{\cap_{j=k-N+1}^k\{\lVert \tilde{\theta}_j \rVert \leq j^{-\varepsilon^\prime}\}} \lVert\tilde{\theta}_{k-N}\rVert^2 \td\mP \nonumber\\
		& + O\left(\ln k \cdot \exp(-mk^{1-\varepsilon})\right) \nonumber\\
		\leq& \left(1+\frac{\beta\delta N L}{k^{1+\varepsilon^\prime}}\right)^2 \mE \lVert\tilde{\theta}_{k-N}\rVert^2 \td\mP + O\left(\ln k \cdot \exp(-mk^{1-\varepsilon})\right).
	\end{align*}
	Substituting the above estimate into \eqref{recur:m2prom_pre} gives
	\begin{equation*}
		\mE\lVert\tilde{\theta}_k\rVert^2\!
		\leq \! \left(1-\frac{N\eta}{k}\right)^2\! \left(1+\frac{\beta\delta N L}{k^{1+\varepsilon^\prime}}\right)^2\! \mE\lVert\tilde{\theta}_{k-N}\rVert^2 + O\left(\frac{1}{k^2}\right).
	\end{equation*}
	Therefore, we have
	\begin{align*}
		&\left.\mE\lVert\tilde{\theta}_k\rVert^2 \middle/ \prod_{i=1}^k\left(1+\frac{\beta \delta N L}{i^{1+\varepsilon^\prime}}\right)^2 \right. \nonumber\\
		\leq& \left(\! 1-\frac{N\eta}{k}\right)^2 \!\! \left. \mE\lVert\tilde{\theta}_{k-N}\rVert^2 \middle/ \prod_{i=1}^{k-N}\left(1+\frac{\beta\delta N L}{i^{1+\varepsilon^\prime}}\right)^2 \right. \!\! + O\left(\frac{1}{k^2}\right).
	\end{align*}
	Then, by \cref{lemma:gen_zhao18}, one can get	
	\begin{equation}
		\left. \mE\lVert\tilde{\theta}_k\rVert^2 \middle/ \prod_{i=1}^k\left(1+\frac{\beta \delta N L}{i^{1+\varepsilon^\prime}}\right)^2 \right.  =
		\begin{cases}
			O\left(\frac{1}{k}\right), & \eta>\frac{1}{2};\\
			O\left(\frac{\ln k}{k}\right), & \eta=\frac{1}{2};\\
			O\left(\frac{1}{k^{2\eta}}\right), & \eta<\frac{1}{2}.
		\end{cases}
	\end{equation}
	Due to the boundedness of $ \prod_{i=1}^\infty\left(1+\frac{\beta \delta N L}{i^{1+\varepsilon^\prime}}\right)^2 $, \eqref{rate:the_ms_prom} is proved.
\end{proof}

\begin{remark}
	By \cref{thm:m2_rate}, the mean square convergence rate of the SA-based algorithm achieves $ O(1/k) $ when properly selecting the coefficient $ \beta $. By \cite{zhang2019asymptotically}, the Cram{\'e}r-Rao lower bound for estimating $ \theta $ based on binary observations $ s_1,\ldots,s_k $ is
	\begin{equation}
		\sigma^2_{CR}(s_1,\ldots,s_k)=\left( \sum_{i=1}^k\frac{f_i^2}{F_i(1-F_i)} \phi_i\phi_i^\top \right)^{-1}=O\left(\frac{1}{k}\right).
	\end{equation}
	Besides, for the identification problem of MA systems with accurate observations and Gaussian noises, the least square algorithm generates a minimum variance  estimate (\!\! \cite{ZHUAN}, Theorem 4.4.2). And, the mean square convergence rate of the recursive least square algorithm is $ O(1/k) $. Therefore, $ O(1/k) $ is the best mean square convergence rate in theory of the identification problem of the {binary-valued} MA systems and even accurate ones.
\end{remark}

\begin{remark}
	In the multiple threshold case, when properly selecting the coefficient $ \beta $, the almost sure and mean square convergence rates of the SA-based algorithm are also $ O(\sqrt{\ln \ln k/k}) $ and $ O(1/k) $, respectively. The analysis is similar to the binary observation case.
\end{remark}

\begin{remark}
	From \cref{thm:as_rate,thm:m2_rate}, we learn that the almost sure and mean square convergence rates are influenced by the step-size, inputs, and the threshold. Here we give intuitive explanations:
	\begin{enumerate}[label={\roman*)}]
		\item The step-size coefficient $ \beta $ influences the convergence rates. If we adopt a small step-size $ \beta $, then the algorithm updates the estimate at a slow rate. 
		\item Excitations of $ \{\phi_k,k\geq 1\} $ also affect the convergence rates. If $ \delta $ in \eqref{AssumEq_SuffRichCond} is large, then $ \{y_k,j\geq 1\} $ provides rich information on $ \theta $ from every direction, which causes good effectiveness of the algorithm. 
		\item The threshold $ C $ is another factor influencing the convergence rates. If the threshold $ C $ is too high or too low, then $ s_k $ may have always the same value, which causes poor effectiveness of the algorithm. 
	\end{enumerate}
	Besides, given $ \{\phi_k, k\geq 1\} $, the upper bound $ M $ does not influence the actual convergence rate of the algorithm, but only influences the estimation on the convergence rates. If $ M $ is too large, then we have less information on $ \{\phi_k,k\geq 1\} $ for estimating the convergence rates, which may lead to an unsatisfactory estimation on the convergence rates.
\end{remark}

\section{Numerical simulation}\label{sec:nume simu}

A numerical simulation will be performed in the section to verify \cref{thm:conv,thm:as_rate,thm:m2_rate}.

Consider an MA system $y_k=\phi_k^\top\theta+d_k$ with binary-valued observation
\begin{equation}
	s_k=I_{\{y_k\leq C\}}=
	\begin{cases}
		1,\quad y_k\leq C;\\
		0,\quad y_k > C,\\
	\end{cases}
\end{equation}
where the unknown parameter $\theta=[3,-1]^\top$, the threshold $C=1$, and $d_k$ is i.i.d. Gaussian noise with variance $\sigma^2=25$ and zero mean. The regressed function of inputs $ \phi_k = [u_k,u_{k-1}]^\top$ is generated by $ u_{3i}=-1+e_{3i},\ u_{3i+1}=2+e_{3i+1},\ u_{3i+2}=1+e_{3i+2} $ for natural number $ i $, where $ e_k = 0.1\sin(\ln(k+1)) $.  
It can be verified that the input follows \cref{a1} with $ M = 2.38 $, $ N = 3 $ and $ \delta = 1.42 $.


In the simulation, set $ \beta = 20 $, $ k_0 = 20 $, and the initial value $\hat{\theta}_{k_0}=[1,1]^\top$.
\cref{fig:simu1} shows a trajectory of $\hat{\theta}_k$. 
\cref{fig:box} gives the box-plots of $ \hat{\theta}_k $ in 200 repeated experiments. 
The figures demonstrate the convergence of Algorithm \eqref{algo}.



\begin{remark}
	We set $ \beta = 20 $ to have $ \eta > 1/2 $. 
	When $ k_0 = 0 $, large $ \beta $ causes large step-sizes in first few steps. 	
	Due to the randomness of $ \{s_k\} $, the estimate $ \hat{\theta}_k $ may run away from the true value $ \theta $ after first few steps of iterations. Then, it will takes much more time to reduce the estimation error. 
	To avoid this situation, we should adjust the starting point $ k_0 $ according the selection of $ \beta $. In the simulation, we set $ k_0 = 20 $.
\end{remark}

\begin{figure}[!htbp]
	\subfigure[A trajectory of $ \hat{\theta}_k $.]{
	\begin{minipage}[t]{1\linewidth}
		\centering
		\includegraphics[width=0.85\linewidth]{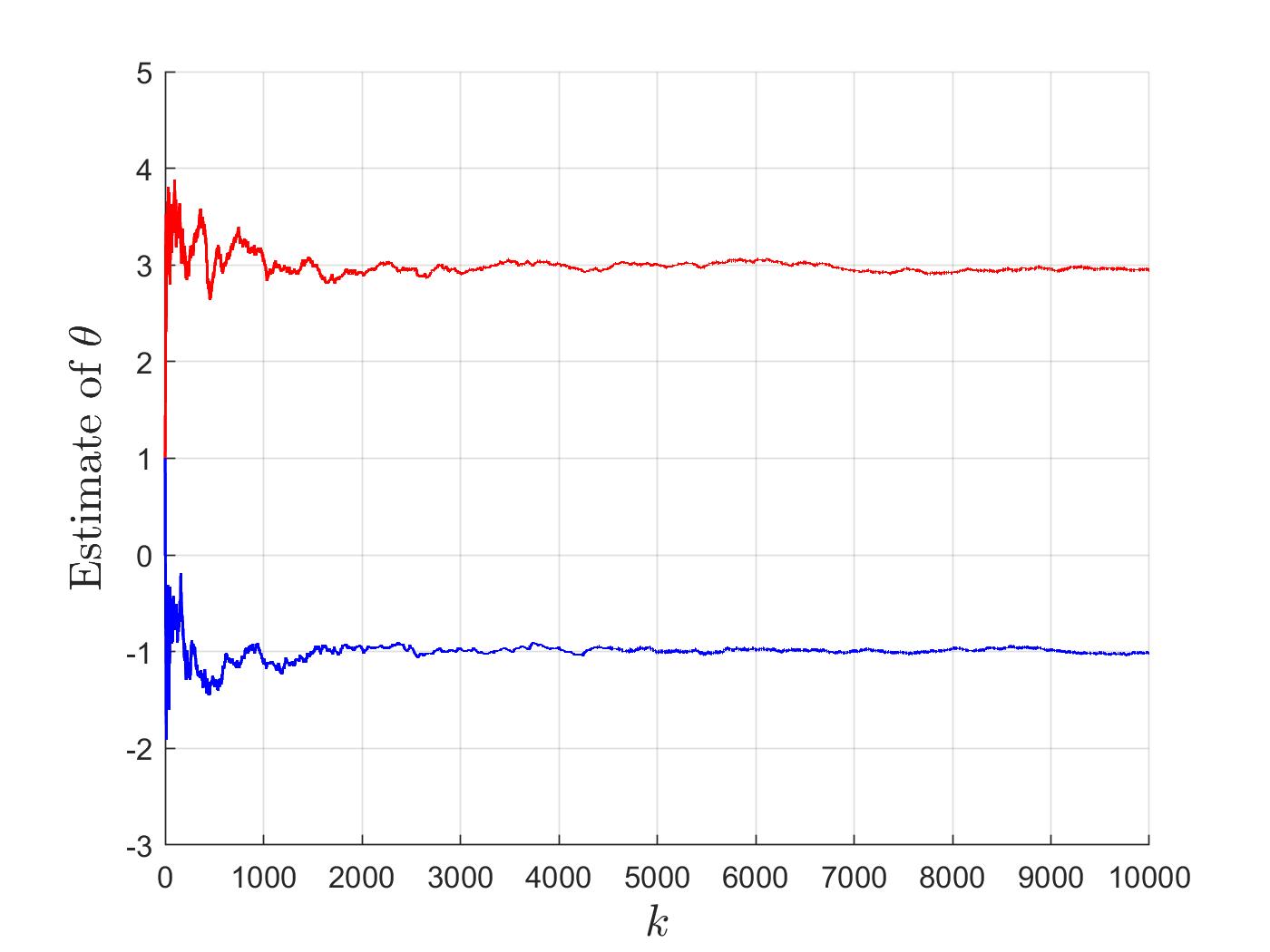}
		\label{fig:simu1}
	\end{minipage}
}
	\subfigure[The box-plots of $ \hat{\theta}_k $ in 200 repeated experiments.]{
	\begin{minipage}[t]{1\linewidth}
		\centering
		\includegraphics[width=0.85\linewidth]{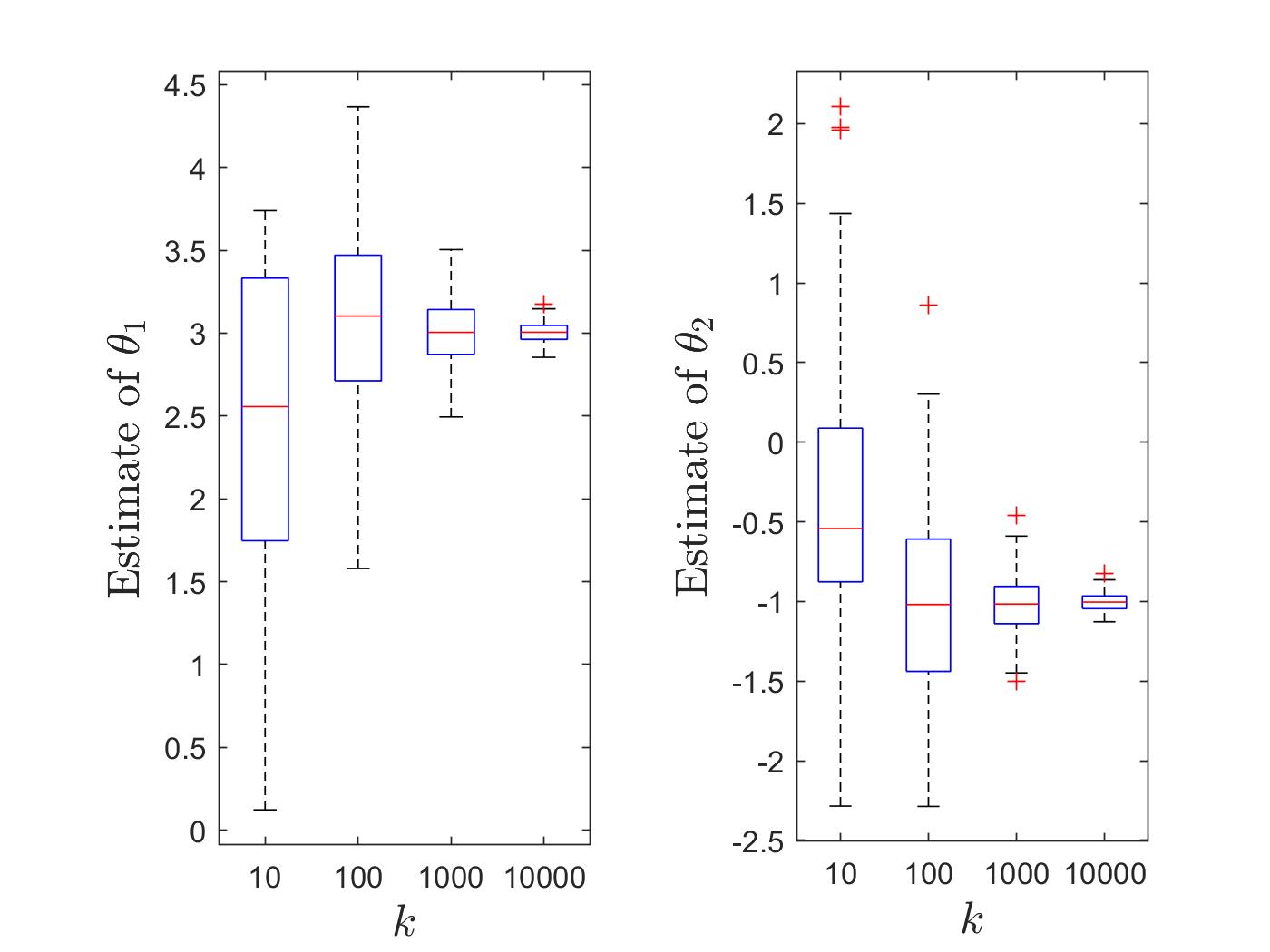}
		\label{fig:box}
	\end{minipage}
}
	\caption{Convergence of Algorithm \eqref{algo}.}
\end{figure}

Note that $ \eta $ is about $ 0.53 $. Then, by \cref{thm:as_rate}, Algorithm \eqref{algo} achieves the almost sure convergence rate of $ O(\sqrt{\ln\ln k/k}) $.  \cref{fig:a.s.convergence rate} shows that the trajectory of $k\lVert \tilde{\theta}_k\rVert^2/\ln \ln k$ is bounded, which consists with the almost sure convergence rate of $ O(\sqrt{\ln\ln k/k}) $.

\begin{figure}[!htbp]
	\centering
	\includegraphics[width=0.85\linewidth]{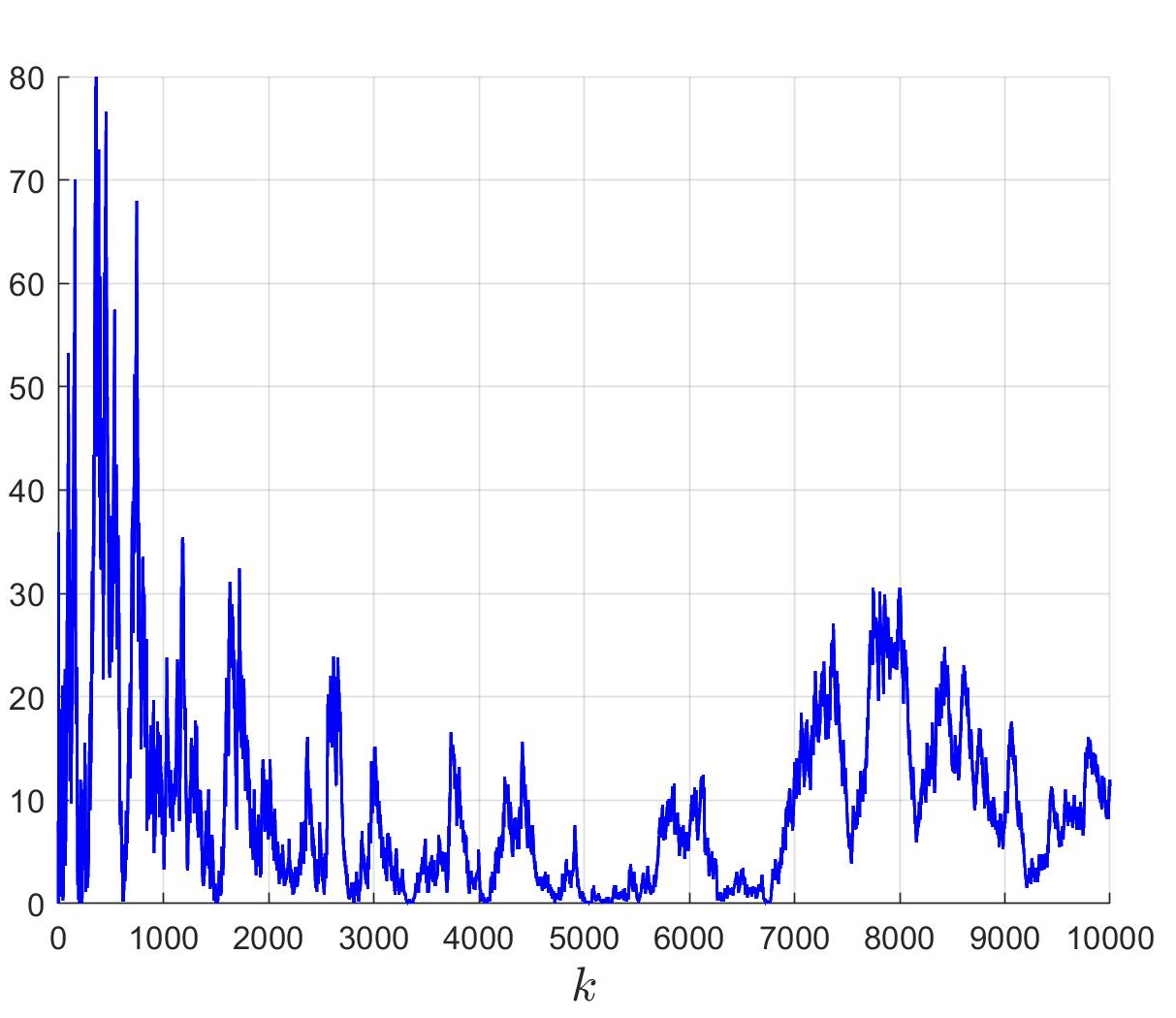}
	\caption{The trajectory of $k\lVert \tilde{\theta}_k\rVert^2/ \ln \ln k$.}
	\label{fig:a.s.convergence rate}
\end{figure}

By \cref{thm:m2_rate}, Algorithm \eqref{algo} achieves the mean square convergence rate of $ O(1/k) $. 
\cref{fig:empirical variance convergence rate} illustrates that the average of the 200 trajectories of $ k\lVert \tilde{\theta}_k\rVert^2 $ is bounded, which consists with the mean square convergence rate of $ O(1/k) $.

Besides, by \cref{thm:m2_rate}, the step-size $ \beta $ influences the mean square convergence rate. \cref{fig:diffeta} shows the empirical mean square convergence rate under the case of $ \beta = 20 $ is faster than that under the case of $ \beta = 1 $. 

{
	It will be an interesting problem to consider the situation where the distribution used in Algorithm \eqref{algo} is different from the actual noise distribution. When the distribution used in Algorithm \eqref{algo} is Gaussian with variance 20 and zero mean, but the actual noise variance is 25, \cref{fig:wrongvariance} shows that the estimation error is bounded. 
}



\begin{figure}[!htbp]
	\centering
	\includegraphics[width=0.85\linewidth]{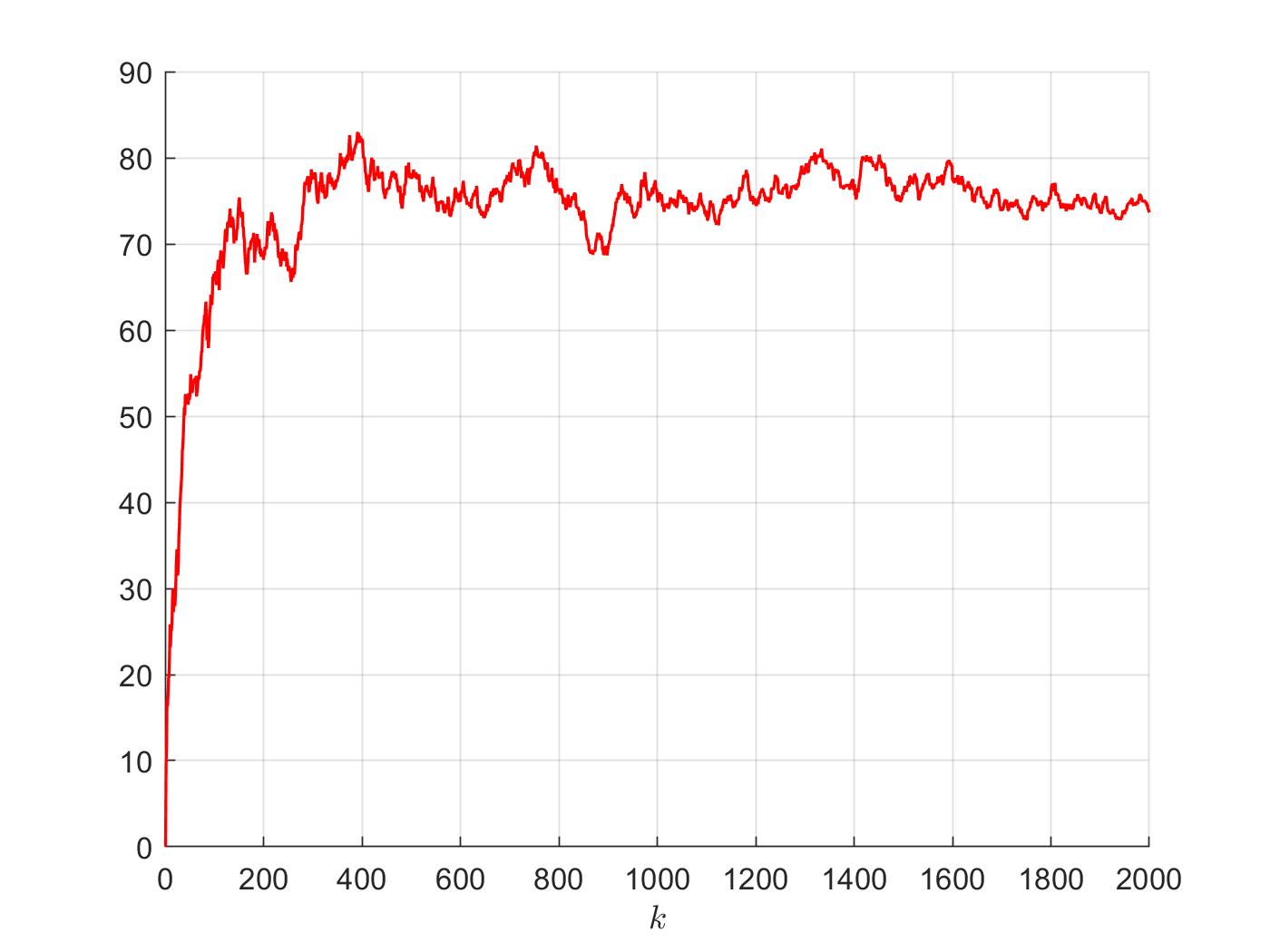}
	\caption{The trajectory of $k\lVert \tilde{\theta}_k\rVert^2$ in 200 repeated experiments.}
	\label{fig:empirical variance convergence rate}
\end{figure}

\begin{figure}[!htbp]
	\centering
	\includegraphics[width=0.85\linewidth]{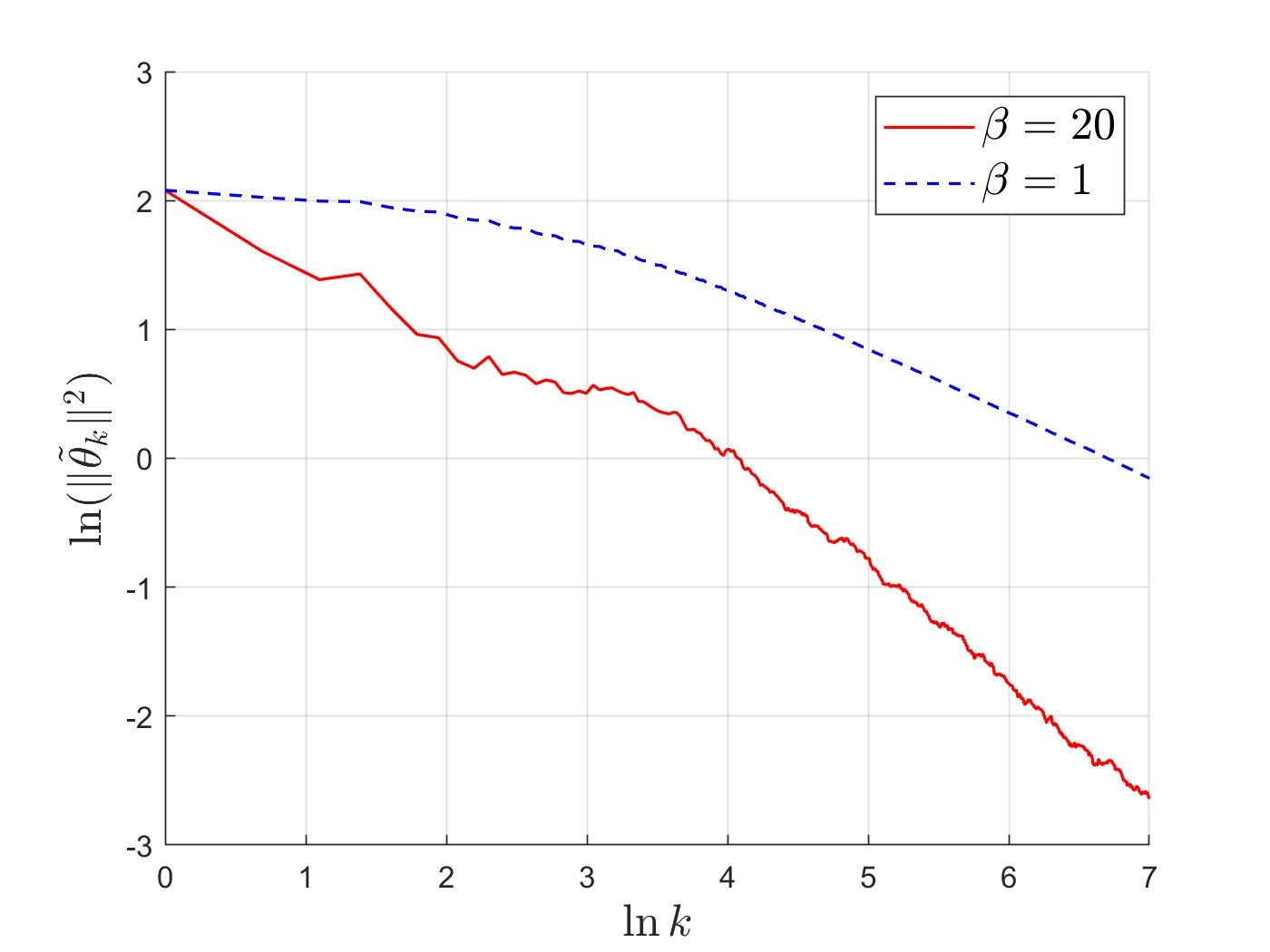}
	\caption{Empirical mean square convergence rates under different $ \beta $.}
	\label{fig:diffeta}
\end{figure}

\begin{figure}[!htbp]
	\centering
	\includegraphics[width=0.85\linewidth]{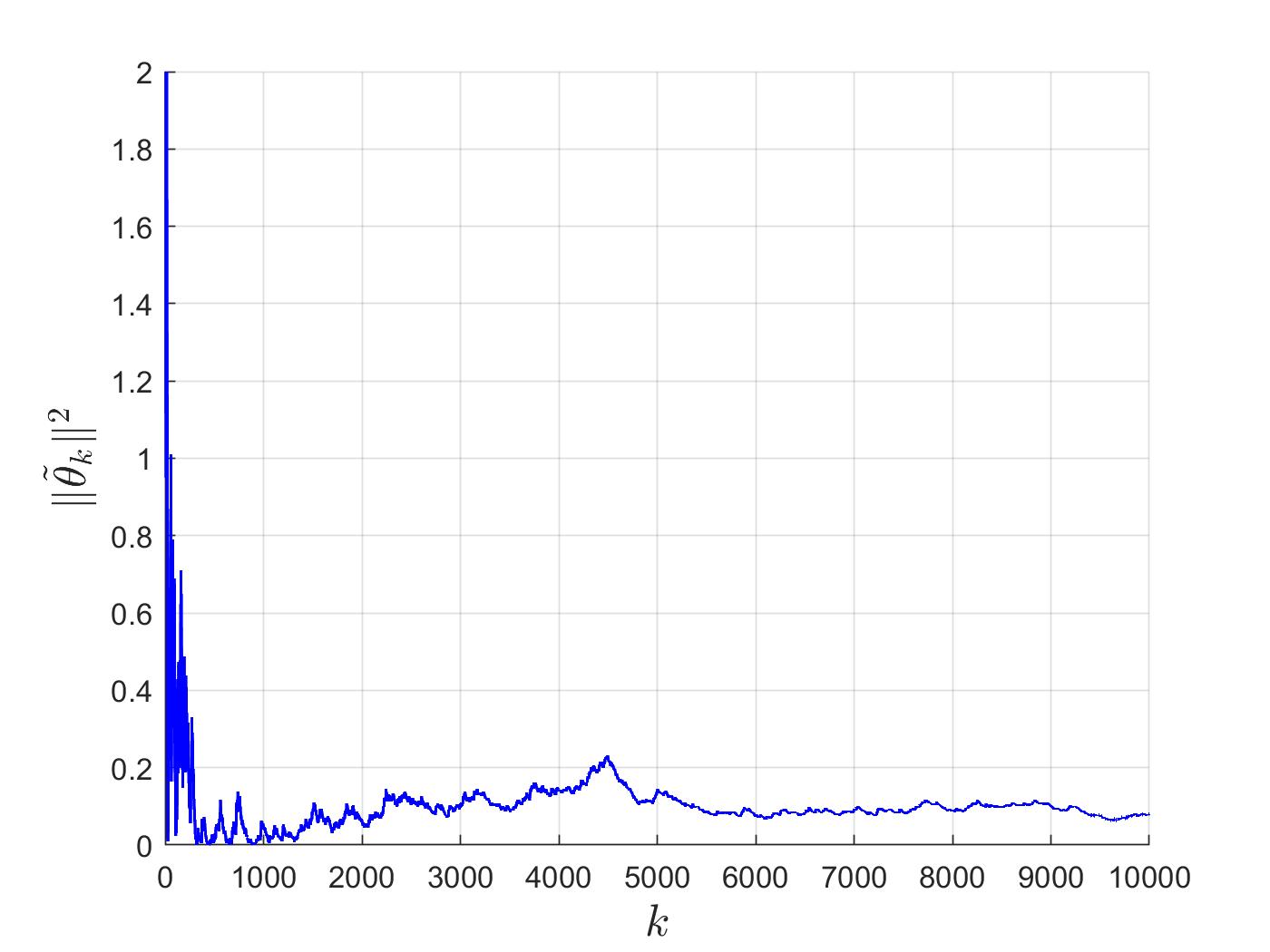}
	\caption{The trajectory of $\lVert \tilde{\theta}_k\rVert^2$ for the wrong variance case.}
	\label{fig:wrongvariance}
\end{figure}


\section{Conclusion}\label{sec:conclusion}

The paper investigates the identification problem of {binary-valued} MA systems with uniformly persistently exciting inputs. An SA-based algorithm without projection is proposed to identify the unknown parameter. The algorithm appears to be the first online identification method for {binary-valued} systems whose implementation does not rely on projections or truncations. 
When properly selecting the coefficients, the almost sure convergence rate of the SA-based algorithm is $ O(\sqrt{\ln \ln k/k}) $, and the mean square convergence rate is $ O(1/k) $. Both the convergence rates are the best for the identification problem of {binary-valued} systems. 
Moreover, an auxiliary stochastic process named SPAO is constructed for the effectiveness analysis. 

Here we give some topics for future research. Firstly, the design of the step-size $\rho_k$ is left as an open question. How can we design a dynamic $\rho_k$ to allow the convergence rates to be the best automatically, and how can we design $\rho_k$ to make the identification algorithm achieve the Cram{\'e}r-Rao lower bound asymptotically? Secondly, can the algorithm be extended to other general forms of systems, such as the infinite impulse response system? And thirdly, how can we design system control laws to regulate the system performance using the SA-based algorithm?


\appendices

\section{Lemmas and the proofs}\label[appen]{appen:proof}

\setcounter{equation}{0}
\renewcommand{\theequation}{A.\arabic{equation}}


\noindent {\it Proof of \cref{lemma:T}.} The lemma can be indicated by Theorem 5.5.1 of \cite{Tucker1967A}. We transfer the problem first.

Firstly, we claim that it is sufficient to prove that there exists $ m>0 $ such that $ \mP\left\{\Abs{w_k}>k^{-\varepsilon}\right\} =O\left(\exp(-mk^{1-2\varepsilon})\right) $. This is because $ \sum_{j=k}^{\infty}\exp(-mj^{1-2\varepsilon})=O\left(k^{2\varepsilon}\exp(-mk^{1-2\varepsilon})\right) =O\left(\exp(-mk^{1-2\varepsilon}/2)\right)$.

Secondly, we claim that it is sufficient to prove that for any $ i\in\{1,2,\ldots,n\} $, $w_{k,i}$ satisfies $ \mP\left\{ \abs{w_{k,i}} > k^{-\varepsilon}/\sqrt{n} \right\} =O\left(\exp(-mk^{1-2\varepsilon})\right) $, where $w_{k,i}$ is the $i$-th component of $w_k$. This is because $\left\{\Abs{w_k}>k^{-\varepsilon}\right\}\subseteq \cup_i \left\{ \abs{w_{k,i}} > k^{-\varepsilon}/\sqrt{n} \right\}$, which implies
\begin{equation*}
	\mP\left\{\Abs{w_k}>k^{-\varepsilon}\right\} \leq \sum_{i=1}^n \mP\left\{ \abs{w_{k,i}} > k^{-\varepsilon}/\sqrt{n} \right\}.
\end{equation*}

The transformation has been finished. And, now we show the converted problem is a corollary of Theorem 5.5.1 of \cite{Tucker1967A}.

\begin{lemmax}[\!\!\cite{Tucker1967A}, Theorem 5.5.1]\label{lemma:Tucker}
	Assume that
	\begin{enumerate}[label=\roman*)]
		\item $\{X_k, k\geq 1\}$ is a sequence of independent random variables;
		\item $\mE X_k = 0$ and $\abs{X_k}\leq \bar{X}<\infty$;
		\item $S_k = \sum_{i=1}^k X_i$, $\sigma_k =\sqrt{{\rm var}(S_k)}$.
	\end{enumerate}
	Then
	\begin{equation}
		\mP\left\{ \frac{S_k}{\sigma_k} > d_k\right\} < \max\left\{\exp\left(-\frac{d_k^2}{4} \right), \exp\left(-\frac{d_k \sigma_k}{4 \bar{X}}  \right) \right\}.
	\end{equation}
\end{lemmax}

Set $X_{k,i}=\beta\phi_{k,i}\left(F_k-s_k\right)$, $ S_{k,i}=\sum_{j=1}^k X_{j,i} $, $ \sigma_{k,i}=\sqrt{\text{var}(S_{k,i})} $ and $d_{k,i}=k^{1-\varepsilon}/\sigma_{k,i}\sqrt{n}$, where $ \phi_{k,i} $ is the  $i$-th component of $ \phi_k $. Then, by \cref{lemma:Tucker},
\begin{align*}
	&\mP\left\{ w_{k,i} > \frac{k^{-\varepsilon}}{\sqrt{n}} \right\}
	= \mP \left\{\frac{S_{k,i}}{\sigma_{k,i}} > d_{k,i} \right\} \nonumber\\
	<&  \max\left\{\exp\left(-\frac{d_{k,i}^2}{4}\right),\exp\left(-\frac{d_{k,i} \sigma_{k,i}}{4 \bar{X}}  \right)  \right\} \nonumber\\
	=& \max\left\{\exp\left( -\frac{k^{2-2\varepsilon}}{4n\sigma_{k,i}^2} \right), \exp\left(-\frac{k^{1-\varepsilon}}{4 \bar{X}\sqrt{n}}  \right)  \right\}.
\end{align*}
Noting that
\begin{equation*}
	\sigma_{k,i}^2=\text{var}(S_{k,i})
	=\sum_{j=1}^k\text{var}(X_{j,i})
	\leq 4\bar{X}^2 k,
\end{equation*}
then $\exp\left( -k^{2-2\varepsilon}/4n\sigma_{k,i}^2\right) \leq \exp\left( -k^{1-2\varepsilon}/16n\bar{X}^2 \right)$. Therefore, there exists  $m_+ >0 $ such that $\mP\left\{ w_{k,i} > k^{-\varepsilon}/\sqrt{n} \right\}=O\left(\exp(-m_+ k^{1-2\varepsilon})\right)$.

$\mP\left\{ w_{k,i} < -k^{-\varepsilon}/\sqrt{n} \right\}$ can be similarly analyzed.

Combining the two consequences, the converted problem is thereby proved. That is to say, we get \cref{lemma:T}. \hfill\qedsymbol

{
	\noindent {\it Proof of \cref{lemma:udf}.} 
	\textbf{(a)} For $ x_1 \geq x_2 $, one can get
	\begin{align*}
		\udf(x_1) = & \sup_{z > M\Abs{\theta}+x_1} \inf\limits_{t\in [C-z,C+z] } f(t) \\
		\leq & \sup_{z > M\Abs{\theta}+x_2} \inf\limits_{t\in [C-z,C+z] } f(t) = \udf(x_2).
	\end{align*}
	Therefore, $ \udf(\cdot) $ is non-increasing. 
	
	Due to the monotonicity of $ \udf(\cdot) $, $ \sup_{x>\chi} \udf(x) $ is the right limit of $ \udf(\cdot) $ at the point $ \chi $. 
	Then, $ \udf(\cdot) $ is right continuous because
	\begin{align*}
		\sup_{x>\chi} \udf(x) 
		= & \sup_{x>\chi} \sup_{z > M\Abs{\theta}+x} \inf_{t\in [C-z,C+z]} f(t) \\
		= & \sup_{z > M\Abs{\theta}+\chi} \inf_{t\in [C-z,C+z]} f(t) = \udf(\chi).  
	\end{align*}
	
	\noindent \textbf{(b)} Since $ \udf(\cdot) $ is right continuous and only defined on $ [0,\infty) $, we have
	$ \lim\limits_{x\to 0} f(x) = \udf(0) = \udf $. 
	
	
	\noindent \textbf{(c)} 
	By \eqref{def:udf_func}, we have
	\begin{align*}
		\udf(x) 
		\leq & \sup_{z \geq M\Abs{\theta}+x} \inf_{t\in [C-z,C+z]} f(t) \\
		= &	\inf\limits_{t\in [C-M\Abs{\theta}-x,C+M\Abs{\theta}+x] } f(t). 
	\end{align*}
	
	\noindent \textbf{(d)} We firstly prove $ g(z) = \inf_{t\in [C-z,C+z]} f(t) $ is locally Lipschitz continuous on $ z \geq 0 $. Since $ f(t) $ is locally Lipschitz continuous, for any given $ z_0 \geq 0 $, there exist $ \delta_1 > 0 $ and $ K_1 > 0 $ such that 
	\begin{align}\label{ineq:lipschitz1}
		\lvert f(t_1) - f(t_2) \rvert \leq K_1 \lvert t_1 - t_2 \rvert
	\end{align}
	for all $ t_1, t_2 \in (C+z_0-\delta_1,C+z_0+\delta_1) $ and $ t_1, t_2 \in (C-z_0-\delta_1,C-z_0+\delta_1) $. 
	
	Consider $ z_1, z_2 \in (z_0 - \delta_1, z_0 + \delta_1) \cap [0,\infty) $. 
	
	If $ z_1 = z_2 $, then $ g(z_2)-g(z_1) = 0 $. 
	
	If $ z_1 \neq z_2 $, then without the loss of generality, consider $ z_1 > z_2 $, which implies 
	$ g(z_1) = \inf_{t\in [C-z_1,C+z_1]} f(t) \leq \inf_{t\in [C-z_2,C+z_2]} f(t) = g(z_2) $. Hence, $ \abs{g(z_1) - g(z_2)} = g(z_2) - g(z_1) $. By the definition of infimum \cite{zorich}, there exists $ \tau_1 \in [C-z_1,C+z_1] $ such that 
	\begin{align}\label{eq:inf_g1}
		g(z_1) = \inf_{t\in [C-z_1,C+z_1]} f(t) \geq f(\tau_1) - ( z_1 - z_2 ).
	\end{align}
	When $ \tau_1 \in [C-z_2,C+z_2] $, 
	\begin{align*}
		g(z_2) - g(z_1) 
		\leq f(\tau_1) - f(\tau_1) + z_1 - z_2 = z_1 - z_2. 
	\end{align*}
	When $ \tau_1 \in [C-z_1,C-z_2) $, set $ \tau_2 = C-z_2 $. Therefore, $ \tau_2 - \tau_1 \leq z_1 - z_2 $, and 
	\begin{align*}
		\tau_1, \tau_2 \subseteq [C-z_1,C-z_2] \subseteq (C-z_0-\delta_1,C-z_0+\delta_1), 
	\end{align*}
	which together with \eqref{ineq:lipschitz1} and \eqref{eq:inf_g1} implies
	\begin{align}\label{ineq:lip_g}
		& g(z_2) - g(z_1)
		\leq  f(\tau_2) - f(\tau_1) + (z_1 - z_2) \nonumber\\
		\leq & K_1 (\tau_2 - \tau_1) + (z_1 - z_2)
		\leq (K_1+1) (z_1 - z_2). 
	\end{align}
	When $ \tau_1 \in (C+z_2,C+z_1] $, set $ \tau_2 = C+z_2 $. Then, \eqref{ineq:lip_g} can be obtained similar to the case of $ \tau_1 \in [C-z_1,C-z_2) $. 
	
	Therefore, $ g(z) $ is locally Lipschitz continuous on $ z \geq 0 $. 
	
	Now we further prove that $ \udf(x) = \sup_{z > M\Abs{\theta}+x} g(z) $ is also local Lipschitz continuous on $ x \geq 0 $. Since $ g(z) $ is locally Lipschitz continuous, for any given $ x_0 \geq 0 $, there exist $ \delta_2 > 0 $ and $ K_2 > 0 $ such that 
	\begin{align}\label{ineq:lipschitz2}
		\lvert g(z_1) - g(z_2) \rvert \leq K_2 \lvert z_1 - z_2 \rvert
	\end{align}
	for all $ z_1, z_2 \in (M\Abs{\theta}+x_0 - \delta_2, M\Abs{\theta}+x_0 + \delta_2) \cap [0,\infty) $. 
	
	Consider $ x_1, x_2 \in (x_0 - \delta_2, x_0 + \delta_2)\cap [0,\infty) $. 
	
	If $ x_1 = x_2 $, then $ \udf(x_1) - \udf(x_2) = 0 $. 
	
	If $ x_1 \neq x_2 $, then without loss of generality, consider $ x_1 > x_2 $, which together with (a) of this lemma implies $ 	\abs{\udf(x_1)-\udf(x_2)} = \udf(x_2)-\udf(x_1) $. By the definition of supremum \cite{zorich}, there exists $ \upsilon_2 \in (M\Abs{\theta}+x_2,\infty) $ such that
	\begin{align}\label{eq:lip_f}
		\udf(x_2) = \sup_{z > M\Abs{\theta}+x_2} g(z) \leq g(\upsilon_2) + (x_1 - x_2). 
	\end{align}
	When $ \upsilon_2 \in (M\Abs{\theta}+x_1,\infty) $, 
	\begin{align*}
		\udf(x_2) - \udf(x_1) \leq g(\upsilon_2)  + (x_1 - x_2) - g(\upsilon_2). 
	\end{align*}
	When $ \upsilon_2 \in (M\Abs{\theta}+x_2,M\Abs{\theta}+x_1] $, set 
	\begin{align*}
		\upsilon_1 = \min\left\{M\Abs{\theta} + 2x_1-x_2,M\Abs{\theta} + \frac{x_1+x_0+\delta_2}{2}\right\}.
	\end{align*}
	Therefore, $ \upsilon_1 > M \Abs{\theta} + x_1 \geq \upsilon_2 $, $ \upsilon_1 - \upsilon_2 < 2(x_1-x_2) $, and 
	\begin{align*}
		\upsilon_1,\upsilon_2 \in (M\Abs{\theta}+x_0 - \delta_2, M\Abs{\theta}+x_0 + \delta_2) \cap [0,\infty),
	\end{align*}
	which together with \eqref{ineq:lipschitz2} and \eqref{eq:lip_f} implies
	\begin{align*}
		& \udf(x_2) - \udf(x_1) \leq g(\upsilon_2) + (x_1-x_2) - g(\upsilon_1) \\
		\leq & K_2( \upsilon_1-\upsilon_2) + (x_1-x_2)
		\leq (2K_2+1) (x_1-x_2). 
	\end{align*}
	Hence, $ \udf(\cdot) $ is local Lipschitz continuous. \hfill\qedsymbol
}

%
%

\begin{lemmax}\label{lemma:Connection}
	Assume that $\phi_k$ satisfies \cref{a1} and the stochastic process $\psi_k$ satisfies $\Abs{\psi_{k}-\psi_{k-1}}\leq \Psi/k$ for some $\Psi>0$. Then,
	\begin{align*}
		\delta \Abs{\psi_{k}}^2 
		\leq& \frac{1}{N}\sum_{j=k+1}^{k+N}\left(\phi_j^\top\psi_{j-1}\right)^2 + \frac{2NM^2\Psi}{k}\sum_{j=k}^{k+N-1}\Abs{\psi_j} \nonumber\\
		& + \frac{N^2M^2\Psi^2}{k^2}.
	\end{align*}
	Furthermore, if $b^\prime>0$ and $k$ is large enough, then there is $k^\prime\in[k+1,k+N]$ such that
	$\lvert \phi_{k^\prime}^\top\psi_{k^\prime-1}\rvert\geq \sqrt{\delta/2}\Abs{\psi_{k}}I_{\{\Abs{\psi_{k}}>b^\prime\}}$.
\end{lemmax}

\begin{proof}
	The lemma is based on \cref{a1}.
	
	Because $\Abs{\psi_{k}-\psi_{k-1}}\leq \Psi/k$, $\Abs{\psi_{k}-\psi_{j-1}}\leq N\Psi/k$ for any $j\in[k+1,k+N]$. Therefore,
	\begin{align*}
		& \frac{1}{N}\sum_{j=k+1}^{k+N}\left(\phi_j^\top\psi_{j-1}\right)^2  \nonumber \\
		\geq & \frac{1}{N}\sum_{j=k+1}^{k+N}\left(\phi_j^\top\psi_{k}\right)^2 - \frac{2NM^2\Psi}{k}\sum_{j=k}^{k+N-1}\Abs{\psi_j} - \frac{N^2M^2\Psi^2}{k^2}.
	\end{align*}
	Besides, by \cref{a1},
	\begin{equation*}
		\! \frac{1}{N}\sum_{j=k+1}^{k+N}\left(\phi_j^\top\psi_{k}\right)^2
		=\frac{1}{N}\sum_{j=k+1}^{k+N}\psi_{k}^\top \phi_j \phi_j^\top \psi_{k}\geq \delta \Abs{\psi_{k}}^2.\!\!\!\!	
	\end{equation*}
	Thus, the first part of the lemma is proved.
	
	As for the second part, we note that under the condition of the lemma, $\psi_k=O(\ln k)$. Then,
	\begin{equation*}
		\frac{2NM^2\Psi}{k}\sum_{j=k}^{k+N-1}\Abs{\psi_j} + \frac{N^2M^2\Psi^2}{k^2} = O\left(\frac{\ln k}{k}\right).
	\end{equation*}
	
	Hence, if $ \Abs{\psi_{k}}>b^\prime $ and $ k $ is sufficiently large, then one can get
	\begin{equation}
		\frac{1}{N}\sum_{j=k+1}^{k+N}\left(\phi_j^\top\psi_{j-1}\right)^2\geq \delta \Abs{\psi_{k}}^2+O\left(\frac{\ln k}{k}\right)>\frac{\delta}{2} \Abs{\psi_{k}}^2,
	\end{equation}
	which implies $ \frac{1}{N}\sum_{j=k+1}^{k+N}\left(\phi_j^\top\psi_{j-1}\right)^2>\frac{\delta}{2} \Abs{\psi_{k}}^2 I_{\{\Abs{\psi_{k}}>b^\prime\}} $ for sufficiently large $ k $. Then, there exists $k^\prime\in\left[k+1,k+N\right]$ such that $\left(\phi_{k^\prime}^\top\psi_{k^\prime-1}\right)^2\geq \frac{\delta}{2} \Abs{\psi_{k}}^2 I_{\{\Abs{\psi_{k}}>b^\prime\}}$, which verifies the second part of the lemma.
\end{proof}

\begin{lemmax}\label{lemma:Seq}
	If a sequence $ \{a_k\} $ satisfies the recursive function
	\begin{equation}\label{recur_LemmaSeq}
		a_{k} \leq a_{k-1}-\frac{D\sqrt{a_{k-1}}}{k+k_0}I_{\{a_{k-1}\geq \frac{\Mpri}{2}\}}+{\nu_k},
	\end{equation}
	where $D$, $k_0$ and $\Mpri$ are all positive, and $\sum_{k=1}^{\infty}\abs{\nu_k}<\Mpri/2$, then
	\begin{equation}\label{Eq_SeqConcl}
		a_k \! < \! \max\left\{ \! \Mpri,\left[\left(\sqrt{a_0} -\frac{D}{2} \ln\left( \frac{k+k_0+1}{k_0+1} \right)\right)^+\right]^2 \!\!\! +\frac{\Mpri}{2} \! \right\},
	\end{equation}
	where $x^+=\max\{0,x\}$.
\end{lemmax}

\begin{proof}
	If $ a_k < \Mpri $, then the lemma is proved. Hence, we can assume that $ a_k\geq \Mpri $ in the rest of the proof, which implies
	\begin{equation*}
		a_t \geq a_k - \sum_{i=t+1}^k \nu_i \geq a_k - \frac{\Mpri}{2} \geq \frac{\Mpri}{2}, \quad \forall t\leq k.
	\end{equation*}
	
	Define $a^\prime_0=a_0$ and $a^\prime_t = a_t - \sum_{i=1}^{t}\abs{\nu_i}>\Mpri/2-\Mpri/2=0$ for $t\geq 1$. Then, we have
	\begin{align*}
		a^\prime_t
		=& a_t - \sum_{i=1}^{t}\abs{\nu_i}
		\leq a_{t-1} - \frac{D\sqrt{a_{t-1}}}{t+k_0} + \nu_t - \sum_{i=1}^{t}\abs{\nu_i} \nonumber\\
		\leq& a_{t-1}- \sum_{i=1}^{t-1}\abs{\nu_i}-\frac{D\sqrt{\left(a_{t-1}- \sum_{i=1}^{t-1}\abs{\nu_i}\right)^+}}{t+k_0} \nonumber\\
		=& a^\prime_{t-1}-\frac{D\sqrt{a^\prime_{t-1}}}{t+k_0},
	\end{align*}
	and hence,
	\begin{align*}
		a^\prime_k
		< & a^\prime_{k-1}-\frac{D\sqrt{a^\prime_{k-1}}}{k+k_0}+\frac{D^2}{4(k+k_0)^2} \nonumber\\
		= & \left(\sqrt{a^\prime_{k-1}}-\frac{D}{2(k+k_0)}\right)^2,
	\end{align*}
	which implies $\sqrt{a^\prime_{k}}< \sqrt{a^\prime_{k-1}}-\frac{D}{2(k+k_0)}$. Therefore, by $ x\leq x^+ $,
	\begin{align*}
		\sqrt{a^\prime_k}
		< & \sqrt{a^\prime_0}- \sum_{t=1}^{k} \frac{D}{2(t+k_0)} \nonumber\\
		\leq & \left(\sqrt{a_0} -\frac{D}{2}\ln\left( \frac{k+k_0+1}{k_0+1} \right)\right)^+.
	\end{align*}
	So, we have
	\begin{align*}
		a_k
		=& a^\prime_k+\sum_{i=1}^{t}\abs{\nu_i} \nonumber\\
		<& \left[\left(\sqrt{a_0} -\frac{D}{2} \ln\left( \frac{k+k_0+1}{k_0+1} \right)\right)^+\right]^2+\frac{\Mpri}{2}.
	\end{align*}
	The lemma is thereby proved.
\end{proof}

\begin{remarkx}
	\cref{lemma:Seq} ensures the uniform ultimate upper boundedness of the sequence $ \{a_k\} $ which satisfies \eqref{recur_LemmaSeq}. Given the initial value $ a_0 $,
	\begin{equation*}
		\sqrt{a_0} -\frac{D}{2} \ln\left( \frac{k+k_0+1}{k_0+1} \right)<0
	\end{equation*}
	when $ k > \left(k_0+1\right) \exp(2\sqrt{a_0}/D)-k_0-1 $, which together with \eqref{Eq_SeqConcl} implies $ a_k<\Mpri $.
\end{remarkx}

\begin{lemmax}\label{lemma:vgx}
	Assume that
	\begin{enumerate}[label=\roman*)]
		\item $ v(\cdot): \mathbb{R}^n\to \mathbb{R}$ is a continuously twice differentiable non-negative function, whose second derivative is bounded;
		\item $ g_k(\cdot): \mathbb{R}^n \to \mathbb{R}^n$ is uniformly bounded;
		\item $ \nabla v(x)^\top g_k(x) $ is uniformly upper bounded, where $ \nabla v(\cdot) $ is the gradient of $ v(\cdot) $;
		\item the positive step-size $ \rho_k\in \mathbb{R} $ satisfies $ \lim\limits_{k\to\infty} \rho_k = 0 $;
		\item $ x_k = x_{k-1} + \rho_k g_k(x_{k-1}) $.
	\end{enumerate}
	Then, $ v(x_k) = O\left(\sum_{i=1}^k \rho_i\right) $.
\end{lemmax}

\begin{proof}
	From
	\begin{align*}
		v(x_k)
		=& v(x_{k-1} + \rho_k g_k(x_{k-1})) \nonumber\\
		=& v(x_{k-1}) + \rho_k \nabla v(x_{k-1})^\top g_k(x_{k-1}) + O\left(\rho_k^2\right) \nonumber\\
		\leq& v(x_{k-1}) + O(\rho_k) \leq O\left(\sum_{i=1}^k \rho_i\right),
	\end{align*}
	we get the lemma.
\end{proof}

\begin{corollary}\label{coro:worst_EstErr}
	Under \cref{a1,a2}, the estimation error of Algorithm \eqref{algo} satisfies $ \tilde{\theta}_k = O\left(\sqrt{\ln k}\right) $.
\end{corollary}

\begin{proof}
	Due to the finite covariance of the noise, by Markov inequality (\!\cite{YSC}, Theorem 5.1.1), when $ t $ goes to $ \infty $,
	\begin{align}
		F\left( C- \phi_k^\top\theta - t \right)
		=& \mP\left\{ d_k < C - \phi_k^\top\theta - t \right\}\nonumber\\
		\leq&  \mP \left\{ d_k^2 > \left( C - \phi_k^\top\theta - t \right)^2 \right\}\nonumber\\
		\leq& \frac{\mE d_k^2}{\left( C - \phi_k^\top\theta - t \right)^2}
		= O\left(\frac{1}{t^2}\right), \label{lim:F_tail1}
	\end{align}
	and similarly, when $ t $ goes to $ -\infty $,
	\begin{equation}\label{lim:F_tail2}
		1- F\left( C- \phi_k^\top\theta - t \right) = O\left(\frac{1}{t^2}\right).
	\end{equation}
	
	Set $ v(x) = x^\top x $. Then, $ \nabla v(x) = x $. By \eqref{lim:F_tail1} and \eqref{lim:F_tail2},
	\begin{align*}
		& \nabla v(x)^\top \phi_k\left( F\left( C- \phi_k^\top\theta - \phi_k^\top x \right) - s_k \right) \nonumber\\
		= & \phi_k^\top x \left( F\left( C- \phi_k^\top\theta - \phi_k^\top x \right) - s_k \right)
	\end{align*}
	is uniformly upper bounded. Thus, we get the corollary by \cref{lemma:vgx}.
\end{proof}

\begin{corollary}\label{coro:worst_SPAO}
	Under the condition of \cref{prop:dist}, $ \psi_k = O(\sqrt{\ln k}) $.
\end{corollary}

\begin{proof}
	From $ \psi_k = \tilde{\theta}_k - w_k = O(\sqrt{\ln k}) + O(1) $, we get the corollary.
\end{proof}

\begin{remarkx}
	\cref{coro:worst_EstErr,coro:worst_SPAO} estimate the estimation error $ \tilde{\theta}_k $ and SPAO $ \psi_k $ in the worst case, respectively.
\end{remarkx}

\begin{lemmax}\label{lemma:gen_zhao18}
	For the sequence $ \{h_k\} $, assume that
	\begin{enumerate}[label=\roman*)]
		\item $ h_k $ is positive and monotonically increasing;
		\item $ \ln h_k = o(\ln k) $.
	\end{enumerate}
	Then, for non-negative real numbers {$ i_0 $}, $ {\eta^\prime} $ and $ \varepsilon $, and any positive integer $ p $,
	\begin{equation*}
		\sum_{l=1}^{k}\prod_{i=l+1}^k\left(1-\frac{{\eta^\prime}}{i+{i_0}}\right)^p\frac{h_l}{l^{1+\varepsilon}}=
		\begin{cases}
			O\left(\frac{h_k}{k^{\varepsilon}}\right), &p{\eta^\prime}>\varepsilon;\\
			O\left(\frac{h_k \ln k}{k^{\varepsilon}}\right), &p{\eta^\prime}=\varepsilon;\\
			O\left(\frac{1}{k^{p{\eta^\prime}}}\right), &p{\eta^\prime}<\varepsilon.
		\end{cases}
	\end{equation*}
\end{lemmax}

\begin{proof}
	{
	Since $ i_0 \geq 0 $, one can get $ \frac{l+1+i_0}{l} = 1 + \frac{i_0+1}{l} \leq  2 + i_0 $ and $ \frac{k}{k+i_0} \leq 1 $ for all positive integers $ l $ and $ k $. 
	Then, by Lemma A.2 in \cite{WangJM2024bi}, we have
	\begin{align*}
		\prod_{i=l+1}^k\left(1-\frac{\eta^\prime}{i+i_0}\right)
		\leq \left(\frac{l+1+i_0}{k+i_0}\right)^{\eta^\prime}
		\leq (2+i)^{\eta^\prime} \left( \frac{l}{k} \right)^{\eta^\prime},
	\end{align*}
	which leads to
	\begin{align*}
		& \sum_{l=1}^{k}\prod_{i=l+1}^k\left(1-\frac{\eta^\prime}{i+i_0}\right)^p\frac{h_l}{l^{1+\varepsilon}}\nonumber\\
		=& \sum_{l=1}^{k}\left[\prod_{i=l+1}^k\left(1-\frac{\eta^\prime}{i+i_0}\right)\right]^p\frac{h_l}{l^{1+\varepsilon}} \\
		\leq & 
		(2+i)^{p\eta^\prime} \sum_{l=1}^{k}\left(\frac{l}{k}\right)^{p\eta^\prime}\frac{h_l}{l^{1+\varepsilon}} 
		= O\left(\frac{1}{k^{p\eta^\prime}}\sum_{l=1}^{k}\frac{h_l}{l^{1+\varepsilon-p\eta^\prime}}\right).
	\end{align*}
}
	Then, it suffices to estimate $ \sum_{l=1}^{k}h_l/l^{1+\varepsilon-p\eta^\prime} $.
	
	Firstly, when $ p\eta^\prime<\varepsilon $, by $ \ln h_k = o(\ln k) $, we have $ h_k<k^{(\varepsilon-p\eta^\prime)/2} $ for sufficiently large $ k $, which implies $ \sum_{l=1}^{\infty}\frac{h_l}{l^{1+\varepsilon-p\eta^\prime}}<\infty $. So, we can get
	\begin{equation*}
		\sum_{l=1}^{k}\prod_{i=l+1}^k\left(1-\frac{\eta^\prime}{i+i_0}\right)^p\frac{h_l}{l^{1+\varepsilon}} =O\left(\frac{1}{k^{p\eta^\prime}}\right).
	\end{equation*}
	
	Secondly, by the monotonicity of $ h_k $, we have
	\begin{align*}
		\sum_{l=1}^{k}\frac{h_l}{l}
		\leq& \sum_{l=1}^{k}h_l\left(\ln l - \ln (l-1) \right)\nonumber\\
		\leq& \sum_{l=1}^{k} \left( h_l \ln l - h_{l-1} \ln (l-1) \right) = h_k\ln k.
	\end{align*}
	Hence, when $ p\eta^\prime=\varepsilon $, one can get
	\begin{equation*}
		\sum_{l=1}^{k}\prod_{i=l+1}^k\left(1-\frac{\eta^\prime}{i+i_0}\right)^p\frac{h_l}{l^{1+\varepsilon}} = O \left(\frac{h_k\ln k}{k^{\varepsilon}}\right).
	\end{equation*}
	
	Lastly, when $ p\eta^\prime>\varepsilon $, we have
	\begin{align*}
		\sum_{l=1}^{k}\frac{h_l}{l^{1 + \varepsilon- p\eta^\prime}}
		=& O\left(\sum_{l=1}^{k} h_l \left( l^{p\eta^\prime-\varepsilon} - (l-1)^{p\eta^\prime-\varepsilon} \right) \right) \nonumber\\
		\leq & O\left(\sum_{l=1}^{k} \left( h_l l^{p\eta^\prime-\varepsilon}- h_{l-1} (l-1)^{p\eta^\prime-\varepsilon}\right) \right) \nonumber\\
		=& O\left( h_k k^{p\eta^\prime-\varepsilon} \right),
	\end{align*}
	which implies
	\begin{equation*}
		\sum_{l=1}^{k}\prod_{i=l+1}^k\left(1-\frac{\eta^\prime}{i+i_0}\right)^p\frac{h_l}{l^{1+\varepsilon}} = O\left(\frac{h_k}{k^{\varepsilon}}\right). \qedhere
	\end{equation*}
\end{proof}

\begin{remarkx}
	If $ h_k $ is constant, $ p=1 $ and $ i_0=0 $, then \cref{lemma:gen_zhao18} implies Lemma 4 in \cite{zhao2018consensus}. Besides, if $ h_k/\ln k $ is assumed to be monotonically decreasing, then the estimate of \cref{lemma:gen_zhao18} is accurate.
\end{remarkx}

\begin{theoremx}\label{prop:prom_psi_tail}
	Under the condition of \cref{prop:dist}, for any $ \varepsilon\in(0,1) $, there exist positive numbers $ \varepsilon^\prime $ and $ m $ such that
	\begin{equation}
		\mP\left\{ \Abs{\psi_k} > k^{-\varepsilon^\prime} \right\} = O\left( \exp\left(-mk^{1-\varepsilon}\right) \right).
	\end{equation}
\end{theoremx}

\begin{proof}
	The theorem can be proved by verifying that there exists $ \varepsilon^\prime>0 $ such that
	\begin{equation}\label{subset:appen}
		\left\{ \Abs{\psi_k} \leq k^{-\varepsilon^\prime} \right\} \supseteq \left\{ \sup_{j\geq \lfloor k^{1-2\varepsilon} \rfloor} j^{\varepsilon}\Abs{w_j}\leq 1 \right\}.
	\end{equation}
	
	By the monotonicity of $ \{ \sup_{j\geq k} j^{\varepsilon}\Abs{w_j}\leq 1 \} $ and \cref{prop:in},
	\begin{equation}
		\left\{ \sup_{j\geq \lfloor k^{1-\varepsilon} \rfloor}\lVert\psi_j\rVert^2 < \Mpri \right\} \supseteq \left\{ \sup_{j\geq \lfloor k^{1-2\varepsilon} \rfloor} j^{\varepsilon}\Abs{w_j}\leq 1 \right\}.
	\end{equation}
	Therefore,  if $ \sup_{j\geq \lfloor k^{1-2\varepsilon} \rfloor} j^{\varepsilon}\Abs{w_j}\leq 1 $, then by \eqref{eq:Lagrange} and \eqref{ineq:matrix_prod}, for all $ j\geq \lfloor k^{1-\varepsilon} \rfloor+N $,
	\begin{equation}
		\Abs{\psi_j}
		\leq  \left( 1 - \frac{\beta \delta N}{j}\udf\left(M\sqrt{\Mpri}\right) \right) \Abs{\psi_{j-N}} + O\left(\frac{1}{j^{1+\varepsilon}}\right), 
	\end{equation}
	where $ \udf(\cdot) $ is defined in \eqref{def:udf_func}. Then, by \cref{coro:worst_SPAO,lemma:gen_zhao18}, $ \Abs{\psi_j} $ converges at a polynomial rate. Hence, we get \eqref{subset:appen}. Then, the theorem can be proved by \cref{lemma:T} and the arbitrariness of $ \varepsilon $.
\end{proof}

\begin{corollary}\label{coro:tail_prom}
	Under the condition of \cref{thm:tail}, for any $ \varepsilon>0 $, there exist positive numbers $ \varepsilon^\prime $ and $ m $ such that
	\begin{equation}
		\mP\left\{ \lVert\tilde{\theta}_k\rVert > k^{-\varepsilon^\prime} \right\} = O\left( \exp\left(-mk^{1-\varepsilon}\right) \right).
	\end{equation}
\end{corollary}

\begin{proof}
	By \eqref{subset:appen} and $ \tilde{\theta}_k = \psi_k + w_k $, we have
	\begin{align*}
		\left\{ \lVert \tilde{\theta}_k \rVert^2 \leq k^{-\varepsilon^\prime} + k^{-\varepsilon} \right\}
		\supseteq &
		\left\{ \Abs{\psi_k} \leq k^{-\varepsilon^\prime} \right\} \cap \left\{ \Abs{w_k}\leq k^{-\varepsilon} \right\} \nonumber\\
		\supseteq & \left\{ \sup_{j\geq \lfloor k^{1-2\varepsilon} \rfloor} j^{\varepsilon}\Abs{w_j}\leq 1 \right\}\!.\!\!\!
	\end{align*}
	Then, the corollary can be proved by \cref{lemma:T}.
\end{proof}

\begin{remarkx}
	\cref{prop:prom_psi_tail,coro:tail_prom} are extensions of \cref{prop:in,thm:tail}, respectively.
\end{remarkx}

\section{Other application of SPAO}\label[appen]{appen:extend}

\setcounter{equation}{0}
\renewcommand{\theequation}{B.\arabic{equation}}

Firstly, the construction of SPAO can be applied to many online identification algorithms of {binary-valued} systems. For {binary-valued} systems with threshold $C_k$, a large number of recursive identification algorithms can be represented as
\begin{equation*}
	\hat{\theta}_k = \hat{\theta}_{k-1} + \rho_k v_k \left(h(\phi_k, \hat{\theta}_{k-1})-s_k \right),
\end{equation*}
where $\{\phi_k, k\geq1\}$ are independent regressed function of inputs, $C_k$ and $v_k$ are generated by $\{\phi_j, s_{j-1}, j\leq k\}$ \cite{Csaji2012recursive,guo2013recursive,Song2018recursive,Song2024identification,Wang2021distributed,You2015recursive,zhang2019asymptotically,WangY2022unified,Fu2022Distributed}. The step-size $\rho_k$ can also be matrices \cite{Wang2021distributed,zhang2019asymptotically}.

Define $ \psi_k = \tilde{\theta}_k - w_k $, where $ \tilde{\theta}_k = \hat{\theta}_k - \theta $ is the estimation error and
\begin{align*}
	w_k = & \rho_k \left(\sum_{i=1}^k v_i\left(\mE\left[s_i | \phi_j,s_{j-1},j\leq i\right]-s_i\right)\right) \nonumber\\
	= & \rho_k \left(\sum_{i=1}^k v_i\left(F(C_i-\phi_i^\top\theta)-s_i\right)\right).
\end{align*}
Then, one can get
\begin{align*}
	\psi_k
	=&\psi_{k-1}+\rho_k\left(\rho_k^{-1}-\rho_{k-1}^{-1}\right)w_{k-1}\nonumber\\
	& +\rho_k v_k\left(h(\phi_k, \psi_{k-1}+w_{k-1}+\theta) -F(C_k-\phi_k^\top\theta)\right) 
\end{align*}

If there is a good convergence property for $ w_k $, then the trajectory of $ \psi_k $ is similar to that of $ \tilde{\theta}_k $ and that of the deterministic sequence
\begin{equation*}
	\overline{\psi}_k = \overline{\psi}_{k-1}+\rho_k v_k\left(h(\phi_k, \overline{\psi}_{k-1}+\theta) -F(C_k-\phi_k^\top\theta)\right).
\end{equation*}
Therefore, we can analyze the convergence property of the algorithm through $ \psi_k $.

Secondly, SPAO technique can be applied in the robustness analysis of Algorithm \eqref{algo}. If the noise distribution used in our algorithm $ F(\cdot) $ is different from the true noise distribution {$ F_{\text{true}}(\cdot) $}, then by the SPAO technique, we can prove that under the condition of \cref{thm:tail}
\begin{equation}
	\varlimsup\limits_{k\to\infty} \lVert \tilde{\theta}_k \rVert^2 \leq M^{\prime\prime}(\Delta_F),\ \text{a.s.},
\end{equation}
where $  \Delta_F = \sup_{x\in \mathbb{R}} \left| F(x) - {F_{\text{true}}(x)} \right| $, and $ M^{\prime\prime}(\cdot) $ is a positive function satisfying $ \lim\limits_{\Delta_F\to 0} M^{\prime\prime}(\Delta_F) = 0 $. The detailed analysis is similar to Theorem 1, and hence, omitted here.

\begin{IEEEbiography}[{\includegraphics[width=1in,height=1.25in,clip,keepaspectratio]{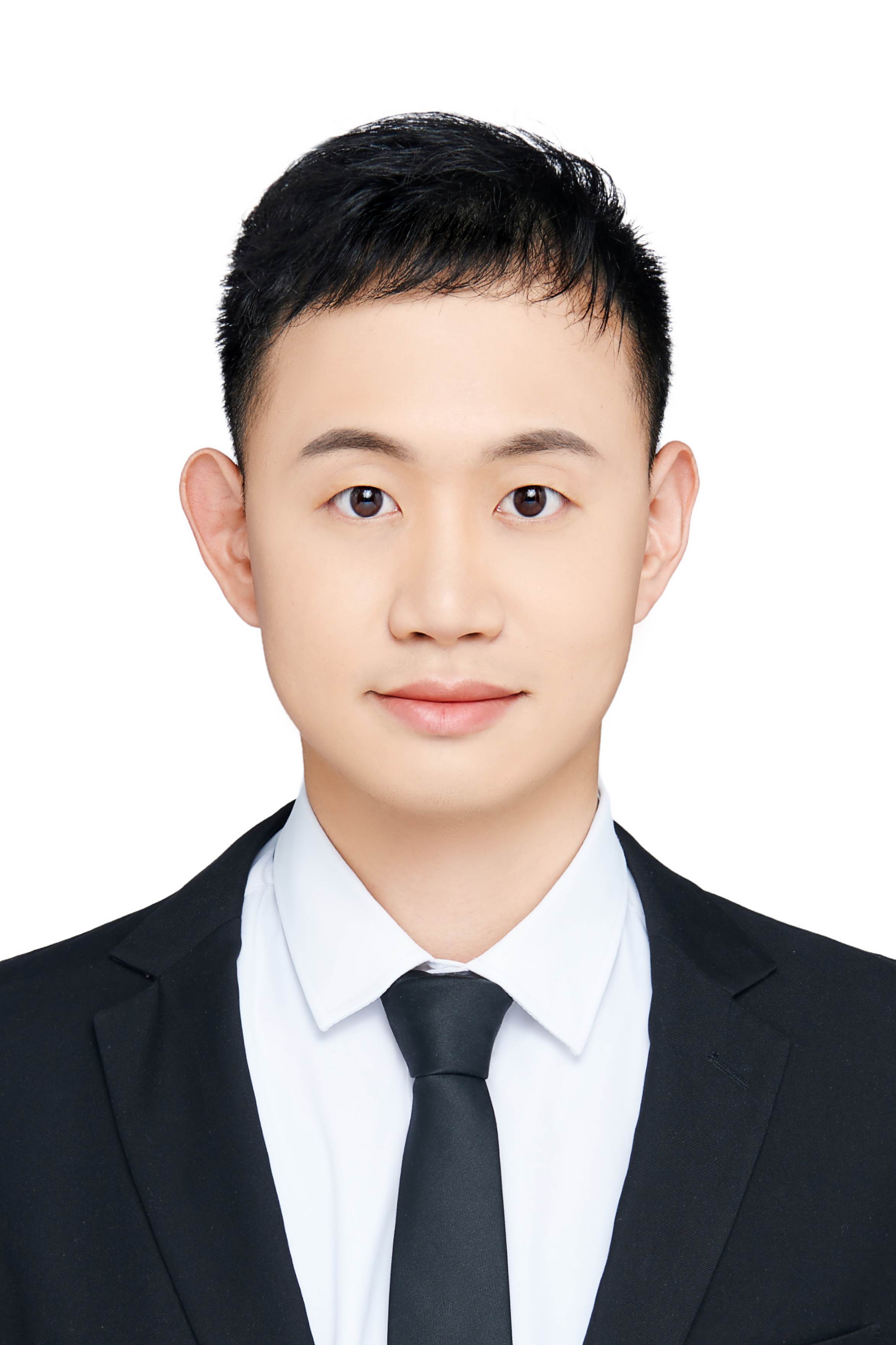}}]{Jieming Ke}(S'22)
	received the B.S. degree in Mathematics from University of Chinese Academy of Science, Beijing, China, in 2020. He is currently working toward the Ph.D. degree majoring in system theory at Academy of Mathematics and Systems Science (AMSS), Chinese Academy of Science (CAS), Beijing, China. 
	
	His research interests include identification and control of quantized systems and the information security problems of control systems. 
\end{IEEEbiography}

\begin{IEEEbiography}[{\includegraphics[width=1in,height=1.25in,clip,keepaspectratio]{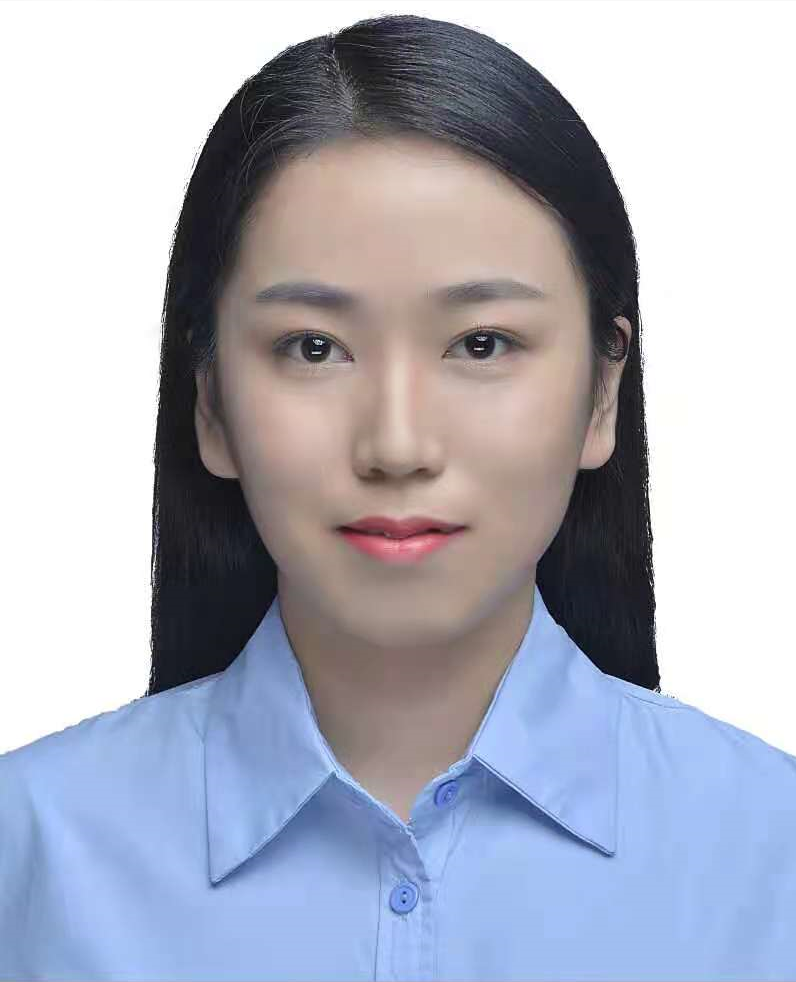}}]{Ying Wang}(S'20-M'22)
	received the B.S. degree in Mathematics from Wuhan University, Wuhan, China, in 2017, and the Ph.D. degree in systems theory from the Academy of Mathematics and Systems Science (AMSS), Chinese Academy of Sciences (CAS), Beijing, China, in 2022. She is currently a Post-Doctoral Research Associate in AMSS, CAS.
	
	Her research interests include identification and control of quantized systems, and multiagent systems. 
\end{IEEEbiography}

\begin{IEEEbiography}[{\includegraphics[width=1in,height=1.25in,clip,keepaspectratio]{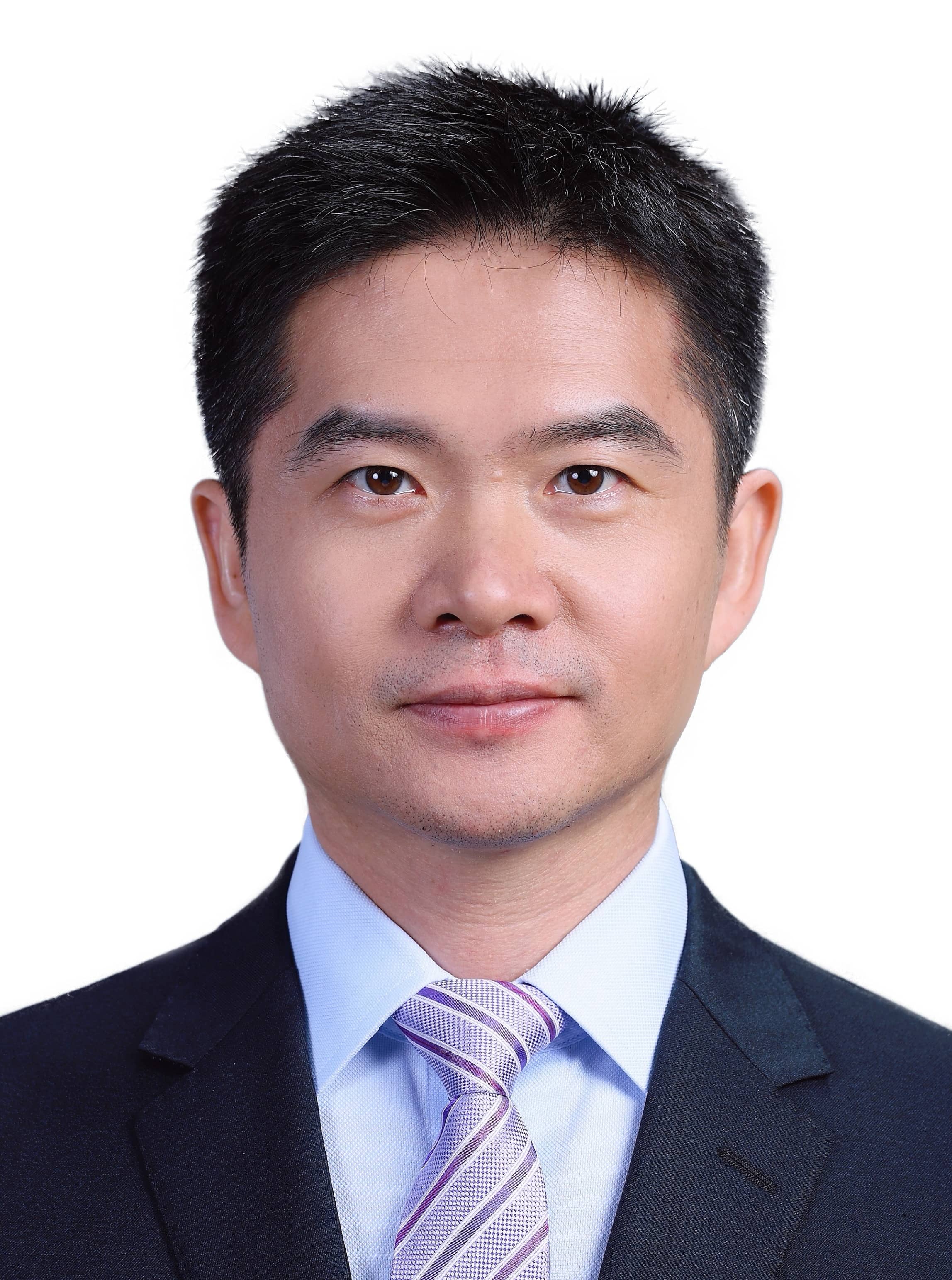}}]{Yanlong Zhao}(S'07-SM'18)
	received the B.S. degree in mathematics from Shandong University, Jinan, China, in 2002, and the Ph.D. degree in systems theory from the Academy of Mathematics and Systems Science (AMSS), Chinese Academy of Sciences (CAS), Beijing, China, in 2007. Since 2007, he has been with the AMSS, CAS, where he is currently a full Professor. His research interests include identification and control of quantized systems, information theory and modeling of financial systems.
	
	He has been a Deputy Editor-in-Chief {\em Journal of Systems and Science and Complexity}, an Associate Editor of {\em Automatica}, {\em SIAM Journal on Control and Optimization}, and {\em IEEE Transactions on Systems, Man and Cybernetics: Systems}. 
	He served as a Vice-President of Asian Control Association and a Vice-President of IEEE CSS Beijing Chapter, and is now a Vice General Secretary of Chinese Association of Automation (CAA) and the Chair of Technical Committee on Control Theory (TCCT), CAA.
\end{IEEEbiography}

\begin{IEEEbiography}[{\includegraphics[width=1in,height=1.25in,clip,keepaspectratio]{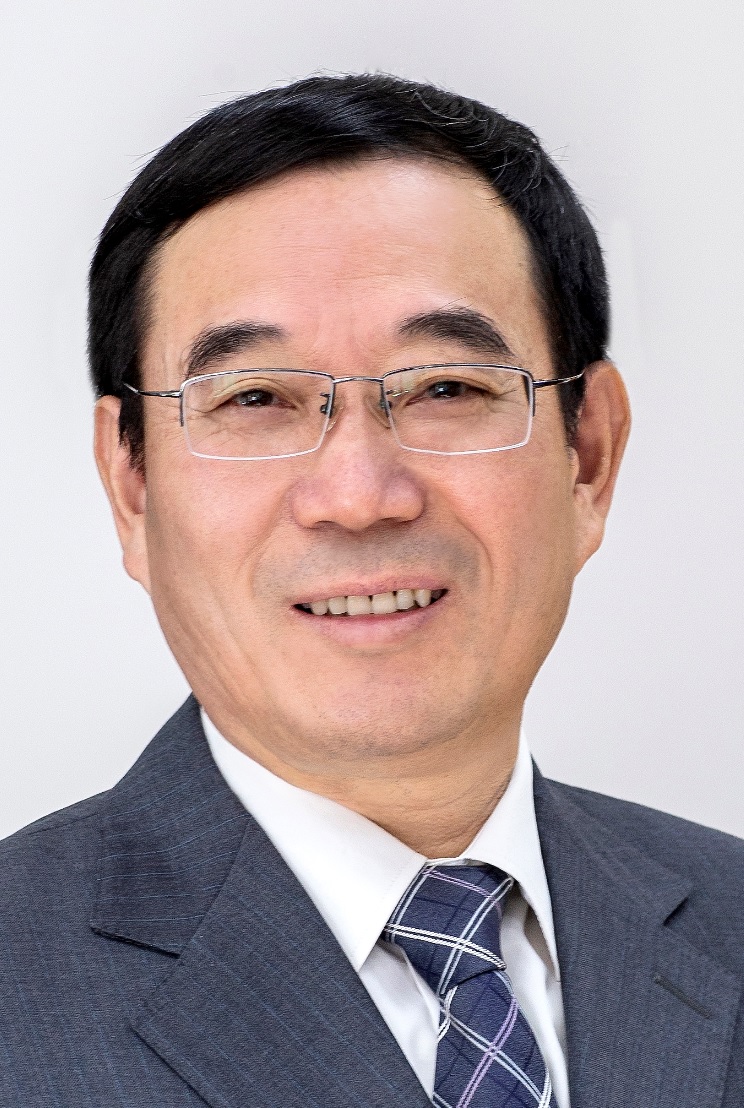}}]{Ji-Feng Zhang}(M'92-SM'97-F'14)
	received the B.S. degree in mathematics from Shandong University, China, in 1985, and the Ph.D. degree from the Institute of Systems Science (ISS), Chinese Academy of Sciences (CAS), China, in 1991. Since 1985, he has been with the ISS, CAS. His current research interests include system modeling, adaptive control, stochastic systems, and multi-agent systems.
	
	He is an IEEE Fellow, IFAC Fellow, CAA Fellow, SIAM Fellow, member of the European Academy of Sciences and Arts, and Academician of the International Academy for Systems and Cybernetic Sciences. He received the Second Prize of the State Natural Science Award of China in 2010 and 2015, respectively. He was a Vice-President of the Chinese Association of Automation, the Chinese Mathematical Society and the Systems Engineering Society of China. He was a Vice-Chair of the IFAC Technical Board, member of the Board of Governors, IEEE Control Systems Society; Convenor of Systems Science Discipline, Academic Degree Committee of the State Council of China. He served as Editor-in-Chief, Deputy Editor-in-Chief or Associate Editor for more than 10 journals, including {\em Science China Information Sciences}, {\em IEEE Transactions on Automatic Control} and {\em SIAM Journal on Control and Optimization} etc.
\end{IEEEbiography}

\end{document}